\documentclass[twocolumn]{IEEEtran}
\IEEEoverridecommandlockouts

\usepackage{amsmath, amssymb, cite}
\usepackage{amsthm,bm,bbm}
\usepackage{mathtools}
\usepackage[small]{caption}
\usepackage{amsthm,multirow,color,amsfonts}
\usepackage{tabulary}
\usepackage{subfigure}
\usepackage{graphicx}
\usepackage{setspace}
\usepackage{enumerate}
\usepackage[]{algorithm2e}
\usepackage{comment}

\DeclareMathOperator*{\Hg}{{\bf H}} 

\newcommand\independent{\protect\mathpalette{\protect\independenT}{\perp}}
\def\independenT#1#2{\mathrel{\rlap{$#1#2$}\mkern2mu{#1#2}}}

\newtheorem{theo}{Theorem}
\newtheorem{defi}{Definition}

\newtheorem{prop}{Proposition}
\newtheorem{remark}{Remark}

\newtheorem{con}{Conjecture}

\newtheorem{ex}{Example}

\newcommand{\derya}[1]{\textit{\textcolor{blue}{ #1}}}

\newcommand{\df}{d_{f_c}(m_v^c)}
\newcommand{\dfLvc}{d_{f_c}(L_v^c)}

\newcommand{\x}{{\bf x}}
\newcommand{\X}{{\bf X}}
\newcommand{\Hgf}{\Hg(f_c(\X))}

\newcommand{\ccomp}{W_{v,comp}^c}
\newcommand{\ccomm}{W_{v,comm}^c}
\newcommand{\ccompt}{W_{v,comp}^{c\,(t)}}
\newcommand{\ccommt}{W_{v,comm}^{c\,(t)}}

\newcommand{\gammaf}{{\bm \gamma}_f({\bm \lambda}^c)}

\newcommand{\Prouc}{P^{rou}(c)}

\newcommand{\Prouctilde}{\tilde{P}^{rou}(c)}

\newcommand{\s}{\mbox{sec}}

\def\compactify{\itemsep=0pt \topsep=0pt \partopsep=0pt \parsep=0pt}
\let\latexusecounter=\usecounter

\newenvironment{CompactEnumerate}
  {\def\usecounter{\compactify\latexusecounter}
   \begin{enumerate}}
  {\end{enumerate}\let\usecounter=\latexusecounter}

\begin{document}
\title{Function Load Balancing over Networks} 
\author{Derya~Malak~and~Muriel~M\'{e}dard
\thanks{Malak is with the ECSE Department, at RPI, Troy, NY 12180 USA (email: malakd@rpi.edu). M\'{e}dard is with RLE, at MIT, Cambridge, MA 02139 USA (email:  medard@mit.edu). \hfill Manuscript last revised: {\today}.}
\thanks{A preliminary version of the paper appeared in Proc., IEEE INFOCOM 2020 \cite{malak2020distribute}.}}
\maketitle

\begin{abstract}
Using networks as a means of computing can reduce the communication flow or the total number of bits transmitted over the networks. In this paper, we propose to distribute the computation load in stationary networks, and formulate a flow-based delay minimization problem that jointly captures the aspects of communications and computation. We exploit the distributed compression scheme of Slepian-Wolf that is applicable under any protocol information where nodes can do compression independently using different codebooks. We introduce the notion of entropic surjectivity as a measure to determine how sparse the function is and to understand the limits of functional compression for computation. We leverage Little'€™s Law for stationary systems to provide a connection between surjectivity and the computation processing factor that reflects the proportion of flow that requires communications. 
This connection gives us an understanding of how much a node (in isolation) should compute to communicate the desired function within the network without putting any assumptions on the topology. Our results suggest that to effectively compute different function classes that have different entropic surjectivities, the networks can be re-structured with the transition probabilities being tailored for functions, i.e., task-based link reservations, which can enable mixing versus separately processing of a diverse function classes.  
They also imply that most of the available resources are reserved for computing low complexity functions versus fewer resources for processing of high complexity ones.
We numerically evaluate our technique for search, MapReduce, and classification functions, and infer how sensitive the processing factor for each computation task to the entropic surjectivity is.
\end{abstract}

\maketitle

\section{Introduction}\label{intro}

Challenges in cloud computing include effectively distributing computation to handle the large volume of data with growing computational demand, and the limited resources in the air interface. Furthermore, various tasks such as computation, storage, communications are inseparable. In network computation is required for reasons of dimensioning, scaling and security, where data is geographically dispersed. 
We need to exploit the sparsity of data within and across sources, as well as the additional sparsity inherent to labeling, to provide approximately minimal representations for labeling.

An equivalent notion to that sparsity is that of data redundancy. Data is redundant in the sense that there exists, a possibly latent and ill understood, sparse representation of it that is parsimonious and minimal, 
and that allows for data reconstruction, possibly in an approximate manner. Redundancy can occur in a single source of data or across multiple sources.

Providing such sparse representation for the reconstruction of data is the topic of compression, or source coding. The Shannon entropy rate of data provides, for a single source, a measure of the minimal representation, in terms of bits per second, required to represent data. This representation is truly minimal, in the sense that it is achievable with arbitrarily small error or distortion, but arbitrarily good fidelity of reconstruction is provably impossible at lower rates. 

\subsection{Motivation}\label{motivation}

As computation becomes increasingly reliant on numerous, possibly geo-dispersed, sources of data, making use of redundancy across multiple sources without the need for onerous coordination across sources becomes increasingly important. The fact that a minimal representation of data can occur across sources without coordination is the topic of distributed compression. The core result is that of Slepian and Wolf \cite{SlepWolf1973}, who showed that distributed compression without coordination 
can be as efficient, in terms of asymptotic minimality of representation. 

Techniques for achieving compression have traditionally relied on coding techniques. Coding, however, suffers from a considerable cost, as it imputes, beyond sampling and quantization, computation and processing at the source before transmission, then computation and processing at the destination after reception of the transmission. A secondary consideration is that coding techniques, to be efficiently reconstructed at the destination, generally require detailed information about the probabilistic structure of the data being represented. For distributed compression, the difficulty of reconstruction rendered the results in \cite{SlepWolf1973} impractical until the 2000s, when channel coding techniques were adapted.


In the case of learning on data, however, it is not the data itself but rather a labeling of it that we seek. That labeling can be viewed as being a function of the original data. The reconstruction of data is in effect a degenerate case where the function is identity. Labeling is generally a highly surjective function and thus induces sparsity, or redundancy, in its output values beyond the sparsity that may be present in the 
data. 
 
The use of the redundancy in both functions and data to provide sparse representations of functions outputs is the topic of the rather nascent field of functional compression. 
A centralized communication scheme requires all data to be transmitted to some central unit in order to perform certain computations. However, in many cases such computations can be performed in a distributed manner at different nodes in the network avoiding transmission of unnecessary information in the network. Hence, intermediate computations can significantly reduce the resource usage, and this can help improve the trade-off between communications and computation.

\subsection{Related Work}
\label{relatedwork}

{\bf \em Distributed compression.} 
Compressed sensing and information theoretic limits of representation provide a solid basis for function computation in distributed environments. Problem of distributed compression has been considered from different perspectives. For source compression, distributed source coding using syndromes (DISCUS) have been proposed \cite{PradRam2003}, and source-splitting techniques have been discussed \cite{ColLeeMedEff2006}. For data compression, there exist some information theoretic limits, such as side information problem \cite{WynZiv1976}, 
Slepian-Wolf coding or compression for multi-depth trees \cite{SlepWolf1973}, and  general networks via multicast and random linear network coding \cite{HoMedKoeKarEffShiLeo2006}. 

{\bf \em Functional compression.} 
In \cite{Korner1973} K\"orner introduced graph entropy 
to characterize rate bounds in functional compression \cite{AlonOrlit1996}. For a general function 
where one source is local and another collocated with the destination, in \cite{OrlRoc2001} authors provided a single-letter characterization of the rate-region. In \cite{DosShaMedEff2010},  
\cite{FeiMed2014} 
authors investigated graph coloring approaches 
for tree networks. In \cite{FES04} authors computed a rate-distortion region for  functional compression with side information, and in \cite{delgosha2018distributed} authors devised polynomial time 
compression techniques for sparse graphs. 
A line of work considered in network computation for 
specific functions. In \cite{KowKum2010} authors studied computing symmetric Boolean functions in tree networks. In \cite{Gal88}, authors analyzed the asymptotic 
rate in noisy broadcast networks, and in \cite{KM08} for random geometric graphs. 
Function computation was studied using multi-commodity flow techniques in \cite{ShaDeyMan2013}. 
However, there do not exist tractable approaches to perform functional compression that approximate the information theoretic limits unlike the case for compression, where coding and compressed sensing techniques exist. 

{\bf \em Computing capacity in networks.} 
Computing capacity of a network code is the maximum number of times the target function can be computed per use of the network 
\cite{HuanTanYangGua2018}. This capacity for special cases such as trees, identity function 
\cite{LiYeuCai2003}, linear network codes to achieve the multicast capacity have been studied \cite{LiYeuCai2003}, \cite{KoeMed2003}. For scalar linear 
functions, the computing capacity can be fully characterized by min cut \cite{KoeEffHoMed2004}. For vector linear functions over a finite field, necessary and sufficient conditions have been obtained so that linear network codes are sufficient to calculate the function \cite{AppusFran2014}. For general functions and network topologies, upper 
bounds on the computing capacity based on cut sets have been studied \cite{KowKum2010}. 
In \cite{HuanTanYangGua2018}, authors generalize the equivalence relation for the computing capacity. 
However, in these works, characterization based on the equivalence relation associated with the target function is only valid for special topologies, e.g., the multi-edge tree. For more general networks, this equivalence relation is not sufficient to explore general computation problems. 

{\bf \em Coded computing and cost-performance tradeoffs.} 
Coding for computing was widely studied in the context of compressed coded computing of single-stage functions in networks and multi-stage computations 
\cite{LiAliYuAves2018} which 
focused on linear reduce functions. 
Nodes with heterogeneous processing capabilities were studied in \cite{KiaWanAves2017}. 
Coded computing aims to tradeoff communications 
by injecting computations. While fully distributed algorithms cause a high communication load, fully centralized systems suffer from high computation load. With distributed computing at intermediate nodes by exploiting multicast opportunities, the communication load can be significantly reduced, and made inversely proportional to the computation load \cite{LiAliYuAves2018}. The rate-memory tradeoff for function computation was studied in \cite{YuAliAves2018}. Other schemes to improve the recovery threshold are Lagrange codes \cite{YuRavSoAve2018}, 
and polynomial codes for distributed matrix multiplication 
\cite{YuAliAve2017}.   

{\bf \em Functions with special structures.} 
In functional compression, functions themselves can also be exploited. There exist functions with special structures, such as sparsity promoting functions \cite{SheSutTri2018}, symmetric functions, 
type sensitive and 
threshold functions \cite{GK05}. One can also exploit a function's surjectivity. 
There are different notions on how to measure surjectivity, such as deficiency \cite{FuaFenWanCar2018}, 
ambiguity 
\cite{PanSakSteWan2011}, and equivalence relationships among function families \cite{gorodilova2019differential}. 

{\bf \em Systems perspective.} Data parallelism 
and model parallelism 
were studied to accelerate the training of convolutional neural networks and multicasting was exploited to reduce communications cost \cite{krizhevsky2014one}. 
Layer-wise parallelism instead of a single parallelism strategy was explored \cite{jia2018exploring}. 
Various scheduling approaches were explored to improve communication efficiency, see e.g., low latency scheduling \cite{ousterhout2013sparrow}, and task graphs for modeling task dependency \cite{jia2018beyond}. 
Another line of work focused on effective task scheduling, e.g., scheduling for cloud computing \cite{boutin2014apollo}, graph partitioning \cite{wang2019supporting} or reinforcement learning to learn efficient assignments on multiple GPUs \cite{mirhoseini2017device}, and placing computation onto a mixture of GPU/CPU devices to reduce the training time 
by modeling device placement as a Markov decision process \cite{gao2018spotlight}. However, \cite{wang2019supporting,mirhoseini2017device,gao2018spotlight} did not explicitly account or optimize communications with respect to the topology. There are also approaches that minimize the node degree while maximizing the number of parallel information exchanges between nodes, such as \cite{wang2019matcha} that considers communications, but in a proportional manner, 
or prioritization of parameter transfers to improve training time, such as \cite{hashemi2018tictac}. 

\subsection{Functional Compression in Networks}
\label{compresstocompute}
In this section, we introduce some concepts from information theory which characterize the minimum communication (in terms of rate) necessary to reliably evaluate a function at a destination. This problem is referred to as distributed functional compression, and has been studied under various forms since the pioneering work of Slepian and Wolf \cite{SlepWolf1973}. 

\paragraph{Slepian-Wolf Compression}
This is the distributed lossless compression setting where the function $f(X_1,\ldots,X_n)$ is the identity function. In the case of two random variables $X_1$ and $X_2$ that are jointly distributed according to $P_{X_1,X_2}$, the Slepian-Wolf theorem gives a theoretical bound for the lossless coding rate for distributed coding of the two statistically dependent i.i.d. finite alphabet source sequences $X_1$ and $X_2$ as shown below \cite{SlepWolf1973}:
\begin{align}
\label{rateregionSW}
R_{X_1} \geq H(X_1|X_2),\quad
R_{X_2} \geq H(X_2|X_1),\nonumber\\
R_{X_1}+R_{X_2} \geq H(X_1,X_2),
\end{align}
implying that $X_1$ can be asymptotically compressed up to the rate $H(X_1|X_2)$ when $X_2$ is available at the receiver \cite{SlepWolf1973}. This theorem states that to jointly recover sources $(X_1,X_2)$ at a receiver with arbitrarily small error probability for long sequences, it is both necessary and sufficient to separately encode $X_1$ and $X_2$ at rates $(R_{X_1}, R_{X_2})$ satisfying (\ref{rateregionSW}). The codebook design is done in a distributed way, i.e., no communication 
is necessary between the encoders.

\begin{figure}[t!]
\centering
\includegraphics[width=0.4\textwidth]{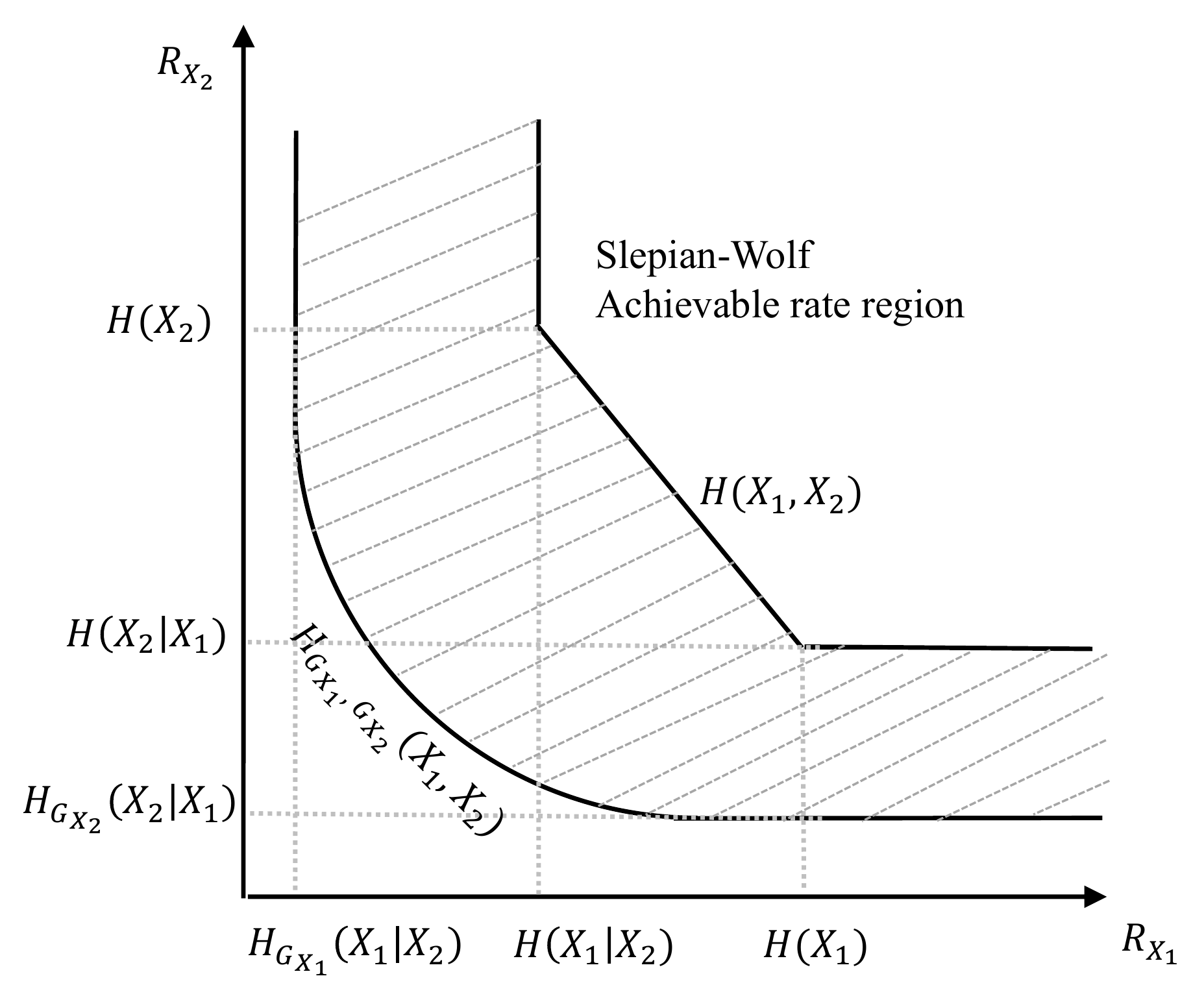}
\caption{Rate region for the zero distortion distributed functional compression problem \cite{DosShaMedEff2010}. 
The joint entropy $H(X_1,X_2)$ is a lower bound to the Slepian-Wolf 
rate region, and the 
joint graph entropy  $H_{G_{X_1},G_{X_2}}(X_1,X_2)$ is characterized by function's surjectivity. 
}
\label{SWandGraphEntropy_surjectivity_details}
\end{figure} 

One challenge in distributed functional compression is the function $f$ on the data $X_i$ itself. The rate region of this problem depends on the function and the mappings from the source variables $X_i$ to the destinations which challenge the codebook design due to the correlations among $X_i$ and $f$. 
To that end, the characteristic graph and its coloring play a critical role in the study of the fundamental limits of functional compression. Each vertex of the characteristic graph represents a possible different sample value, and two vertices are connected if they should be distinguished. More precisely, for a collection of sources $X=\{X_1,\ldots,X_n\}$ each taking values in the same alphabet $\mathcal{X}$, and a function $f$, to construct  the characteristic graph of $f$ on $X_1$, i.e., $G_{X_1}$, we draw an edge between vertices $u$ and $v \in \mathcal{X}$, if $f(u, x_2,\ldots,x_n) \neq f(v, x_2,\ldots, x_n)$ for any $x_2,\ldots,x_n$ whose joint instance has non-zero measure. Surjectivity of $f$ determines the connectivity and hence coloring of the graph. 
The entropy rate of the coloring of $G_X$ for the function $f$ on $X$  
characterizes the minimal representation needed to reconstruct with fidelity the desired function of the data $f(X)$ \cite{Korner1973}.  The degenerate case of the identity function corresponds to having a complete characteristic graph.

\paragraph{Functional Compression}
Given a graph $G_{X_1} = (V_{X_1} , E_{X_1} )$ of $X_1$ with respect to $X_2$, $p(x_1, x_2)$, and function $f(X_1, X_2)$
the graph entropy is expressed as \cite{Korner1973,AlonOrlit1996}
\begin{align}
H_{G_{X_1}}(X_1)= \min\limits_{X_1\in W_1\in S(G_{X_1})} I(X_1; W_1), 
\end{align}
where $S(G_{X_1})$ is the set of all maximal independent sets (MISs) of $G_{X_1}$, where a MIS is not a subset of any other independent set, i.e., a set of vertices in which no two vertices are adjacent \cite{moon1965cliques}, and $X_1 \in W_1 \in S(G_{X_1})$ means that the minimization is over all distributions $p(w_1, x_1)$ such that $p(w_1,x_1) > 0$ implies $x_1 \in w_1$, where $w_1$ is a MIS of $G_{x_1}$.

In \cite[Theorem 41]{FeiMed2014}, authors have characterized the rate region 
for a distributed functional compression problem with two transmitters and a receiver 
by the following conditions:
\begin{align}
\label{rateregiongraph}
R_{11} \geq H_{G_{X_1}}(X_1|X_2),\quad
R_{12} \geq H_{G_{X_2}}(X_2|X_1),\nonumber\\
R_{11}+R_{12} \geq H_{G_{X_1},G_{X_2}}(X_1,X_2),
\end{align}
where $H_{G_{X_1},G_{X_2}}(X_1,X_2)$ is the joint graph entropy. 
To summarize, the role of functional compression is to reduce the amount of rate needed to recover the function $f$ on data $X$, 
and the savings increase along 
$H(X) 
\to H_{G_X}(X)$.

In general, finding minimum entropy colorings of characteristic graphs is NP-hard, and 
the optimal rate region of functional compression remains an open problem \cite{GamKim2011}. However, in some instances, it is possible to efficiently compute these colorings \cite{DosShaMedEff2010}, \cite{FeiMed2014}. 
In \cite{FeiMed2014}, the sources compute colorings of high probability subgraphs of their $G_X$ and perform Slepian-Wolf compression on these colorings and send them. Intermediate nodes compute the colorings for their parents', and by using a look-up table to compute the corresponding functions, they find corresponding source values of received colorings.

In Fig. \ref{SWandGraphEntropy_surjectivity_details}, we illustrate the Slepian-Wolf compression rate region inner bound (IB) in (\ref{rateregionSW}) versus the convex outer bound (OB) in (\ref{rateregiongraph}) determined by the joint graph entropy. 
The region between two bounds 
determines the limits of the functional compression 
and indicates that there could be potentially a lot of benefit in exploiting the compressibility, e.g., via using the deficiency metric in \cite{FuaFenWanCar2018}, of the function to reduce communication. The convexity of OB can be used to exploit the tradeoff between communications and computation, which is mainly determined by the network, data, correlations, and functions. 
%
%
While this gives insights on the limits of compression, it is not clear whether operating on the OB jointly optimizes communications and computation. In particular, since the achievable schemes are based on NP-hard concepts, constructing optimal compression codes imposes a significant computational burden on the encoders and decoders. If the cost of computation were insignificant, it would be optimal to operate at OB. 
When this cost is not negligible, to capture the balance between communications and computation, we will follow a flow-based approach to account for functional compression provided by fundamental entropy limits, which we detail in Sect. \ref{networkmodel}.

\subsection{Contributions}	
\label{contributions}
In this paper, we provide a fresh look at the distributed function computation problem in a networked context. To the best of our knowledge, there are no constructive approaches to this problem except for special cases as outlined in this section. 
As a first step to ease this problem, we will provide a utility-based approach for general cost functions. As special cases, we continue with simple examples of point search with running time $O(\log N)$ where $N$ is the input size in bits, 
then MapReduce with $O(N)$, then the binary classification model with $O(\exp(N))$. Our main contribution is to provide the link between the computation problem and Little's Law.

We devise a distributed function computation framework for stationary Jackson network topologies as a simple means of exploiting function's entropic surjectivity -- a notion of sparsity inherent to functions, by employing the concepts of graph entropy, to provide approximately minimal representations for computing. Function outcomes can be viewed as colors on characteristic graph of the function on the data, which is central to functional compression. Our main insight is that, the main characteristics required for operating the distributed computation scheme are those associated with the entropic surjectivity of the functions. 
Jackson networks allow for the treatment of each node in isolation, independent of the topology. Hence, we do not have to restrict ourselves to cascading operations as in \cite{FeiMed2010allerton} due to the restriction of topology to linear operations.
While the performance analysis of the Jackson network is standard by applying Little's Law and Markov routing policy, and the communication time can be devised easily, to the best of our knowledge, the computing perspective was never integrated to Little's Law before.

We frame a delay cost optimization problem for a Jackson network topology for the distributed function computation by jointly considering the computation and communication aspects and using a flow-based technique where we use general cost functions for computation. 
We introduce entropic surjectivity as a measure to determine how surjective a function is. 
The enabler of our approach is the connection between Little's Law and proportion of flow that requires communications 
that is determined by the entropic surjectivity of functions.

The advantages of the proposed approach is as follows. It does not put any assumptions on the network topology and characterizes the functions only via their entropic surjectivity. 
The probabilistic reservation of the bandwidth (i.e., the Markov routing policy) is on the tasks that have different entropies from a functional compression perspective. The distributed data compression scheme of Slepian-Wolf focuses on the source compression and is applicable under any protocol information \cite{cover1975proof}. 
Our approach provides insights into how to distribute computation only via the entropic surjectivity of the functions, and how to allocate the resources among different functions when flows are processed separately or mixed. 
Our results suggest that to effectively compute different function classes, the networks can be restructured with the transition probabilities being tailored for tasks,  i.e., task-based or weighted link reservations, which can enable mixing of a diverse classes. They also imply that most of the available resources are reserved for computing low complexity functions versus fewer resources for complex functions.

The organization for the rest of the paper is as follows. In Sect. \ref{networkmodel}, we detail the computation model, and derive lower bounds on the rate of generated flows (i.e., processing factors) 
by linking the computation problem to Little's Law. In Sect. \ref{costbreakdown}, we provide the delay cost models. In Sect. \ref{flowanalysis} we analyze the flow, derive load thresholds for computing, and devise cost optimization problems. In Sect. \ref{performance}, we present numerical results, and in Sect. \ref{conclusion} discuss possible directions.

\section{A Networked Computation Model}
\label{networkmodel}
In this section, we detail our network model that incorporates communications and computation.  
We consider a general stationary network topology. Sources can be correlated, and computations are allowed at intermediate nodes, and we are interested in computing a set of deterministic functions. While doing so our goal is to effectively distribute computation. To that end, intermediate nodes need to decide whether to compute or relay.

{\bf \em A product-form multi-class network model.} 
We assume a multiple-class open Jackson network of 
$|V|$ nodes with computing capabilities, in which packets can enter the system from an external/virtual source $s$, get processed and leave the system, i.e., are sent to an virtual destination $0$. Let $G=(V,E)$ be a graph modeling the communication network where $V=[v_0=s,v_1,\hdots, v_{|V|-2}, v_{|V|-1}=0]$ represents the collection of all $|V|$ nodes 
and there is a directed link from $v\in V$ to $w\in V$ if $(v, w)\in E$. We also let $V'=V\backslash\{s,0\}$. 
Multi-class networks of queues and their generalizations to quasi-reversible networks, networks in product form, or networks with symmetric queues have been well explored in the literature, see for example \cite{kelly2011reversibility,nelson2013probability}. 
These classes of networks of queues can capture a broad array of queue service disciplines, including FIFO, LIFO, and processor sharing.


A Jackson network exhibits a product-form equilibrium behavior/distribution such that we can consider each node in isolation (i.e., independently) and investigate the steady-state properties of the network. A benefit of this is each node needs to know how much it needs to manage, which is less complicated than when nodes need the topological information to determine how to manage individual computational flows. 
Similar behavior is observed when a network of queues, such that each individual queue in isolation is quasi-reversible\footnote{A queue with stationary distribution $\pi$ is quasi-reversible if its state at time $t$, ${\bf x}(t)$ is independent of (i) the arrival times for each class of packet subsequent to time $t$, and (ii) the departure times for each class of packet prior to time $t$, for all classes of packet \cite[Ch. 10.3.7]{nelson2013probability} \cite[pp. 66-67]{kelly2011reversibility}.}, always has a product from stationary distribution \cite{walrand1983probabilistic}. Different from reversibility, a stronger condition is imposed on arrival rates and a weaker condition is applied on probability fluxes. 

In our network setting each node maintains a computation queue and a communications queue. 


{\bf \em State of communication / computation 
queue.}  We consider a network $\mathcal{N}$ ($\mathcal{M}$) of $v\in V$ quasi-reversible communications (computation) queues. Let $n_v^c$ ($m_v^c$) be the number of packets of class $c\in C$ at node $v$ requiring communications (compute) service, and $n_v=\sum\nolimits_{c\in C}n_v^c$ ( $m_v=\sum\nolimits_{c\in C}m_v^c$) be the total number of packets in the communications (compute) queue of $v$. To model the communication (computation) queue state of $v$, let $\mathbf{n}_v= (c_1, \hdots , c_{n_v})$ (${\bf m}_v= (c_1, \hdots , c_{m_v})$) where $c_i$ is the class of $i=1,\hdots, n_v$ ($i=1,\hdots, m_v$) th packet. 
A class $c$ packet has an arrival that follows a Poisson process with rate $\lambda_v^c$, and requires a compute service. The traffic intensity associated with class $c$ at $v$ is $\sigma_v^c$. Hence, the total intensity at node $v$ equals $\sigma_v=\sum\nolimits_{c\in C}\sigma_v^c$. Utilization is $\sigma_v^c<1$ to ensure stability of computation.
A class $c$ packet has a departure that follows a Poisson process with rate $\gamma_v^c(\lambda_v^c)$ 
and requires a communication service drawn from an exponential distribution with mean $1/\mu_v^c$ independently. Hence, the traffic intensity associated with class $c$ at $v$ is $\rho_v^c=\gamma_v^c(\lambda_v^c)/\mu_v^c$. Hence, the total intensity at $v$ equals $\rho_v=\sum\nolimits_{c\in C}\rho_v^c$. The utilization satisfies $\rho_v^{c}<1$ for all $c$, $v$, which suffices for stability. 
To ensure quasi-reversibility, we assume a M/M/1 network with $\mu_v^c=\mu_v$ for all $c$. Then the steady-state distribution of the queue $v$ has a product form \cite[Ch. 9.9.2]{SrikantYing2014} given by 
$\pi_v(\mathbf{n}_v)=(1-\rho_v)\prod\nolimits_{c\in C} (\rho_v^{c})^{n_v^c}$. 
The steady-state distribution of the global state of $\mathcal{N}$ is $\pi(\mathbf{n})=\prod \nolimits_{v\in V} \pi_v(\mathbf{n}_v)$, where $\mathbf{n}=(\mathbf{n}_1,\,\mathbf{n}_2,\hdots, \mathbf{n}_{|V|})$. 
The global state of $\mathcal{M}$ is ${\bf m} = ({\bf m}_1,\,{\bf m}_2,\,\hdots,{\bf m}_{|V|})$.   The stationary distribution of compute queue $v$ can be analyzed in isolation similarly as in communications queue 
and is $\phi_v({\bf m}_v)=(1-\sigma_v^c)\prod\nolimits_{c\in C} (\sigma_v^{c})^{m_v^c}$. Hence, the stationary distribution of $\mathcal{M}$ is given by $\phi({\bf m})=\prod\nolimits_{v\in V}\phi_v({\bf m}_v)$. 
Global state of the overall network is $[\mathbf{n}; \mathbf{m}] =[\mathbf{n}_v; \mathbf{m}_v]_{v\in V}\in\mathbb{R}^{n+m}$ where $n=\sum\nolimits_{v\in V}{n_v}$, and $m=\sum\nolimits_{v\in V}{m_v}$. 
For simplicity of notation, we let $\rho_v^c={\lambda_v^c}/{\mu_v^c}\in[0,1)$, and $\rho=[\rho_v^c]_{c\in  C,\,v\in V}$, $\lambda=[\lambda_v^c]_{c\in  C,\,v\in V}$, $\mu=[\mu_v^c]_{c\in  C,\,v\in V}$, and $\sigma=[\sigma_v^c]_{c\in  C,\,v\in V}$.

{\bf \em Set of computation flows.}
For a set of functions $f\in\mathcal{F}$, we denote the set of computational flows by $ C = \{c = (f, X) \}$ which represents the class of functions and $X\in\mathcal{X}$ defined on the probability space $(\mathcal{X},\mathcal{P})$ where $\mathcal{X}$ is the set of symbols and $\mathcal{P}$ is the data (or source) distribution. Multiple classes represent different types of functions, computation for each class is done in parallel, and the network can do computations and conversions between different classes (which we detail in Sect. \ref{costbreakdown}). Let $ C=|C|$ be the cardinality of all classes, and any function of the same class $c\in C$ has the same complexity. Hence, we denote a function $f$ by $f_c$ to if it belongs to class $c$. How many classes of functions that the network can handle will be a proxy for the resolution of network's computation capability, and $|C|$ being sufficiently large means that a wide class of functions can be computed/approximated.

{\bf \em Source data in bits.}  
A stochastic process models the source $\X=(X_1,\hdots X_N)$. Hence, the entropy rate or source information rate is given as $H(X)=\lim\limits_{N\to \infty} H(\X)/N$ assuming that the limit exists. We assume that the total arrival rate (flow in bits per second) of $\X$ at node $v$ is $\lambda_v=H(X)$ and the total incoming traffic intensity associated with class $c$ at node $v$ is $\lambda_v^c$, which is always less than the service rate of class $c$ flow at node $v$ given by $\mu_v^c$. Average number of packets at node $v$ due to the processing of class $c$ function is $L_v^c$.

\begin{defi}{\bf Processing factor.} This 
is the amount of computational flow rate generated by a node $v$ as a result of computing $f_c$ and is a monotone increasing function of the incoming flow rate $\lambda_v^c$. It is smaller than $\lambda_v^c$ and denoted by $$\gamma_f(\lambda_v^c).$$ 
\end{defi}

\begin{figure}[t!]
\centering
\includegraphics[width=0.75\columnwidth]{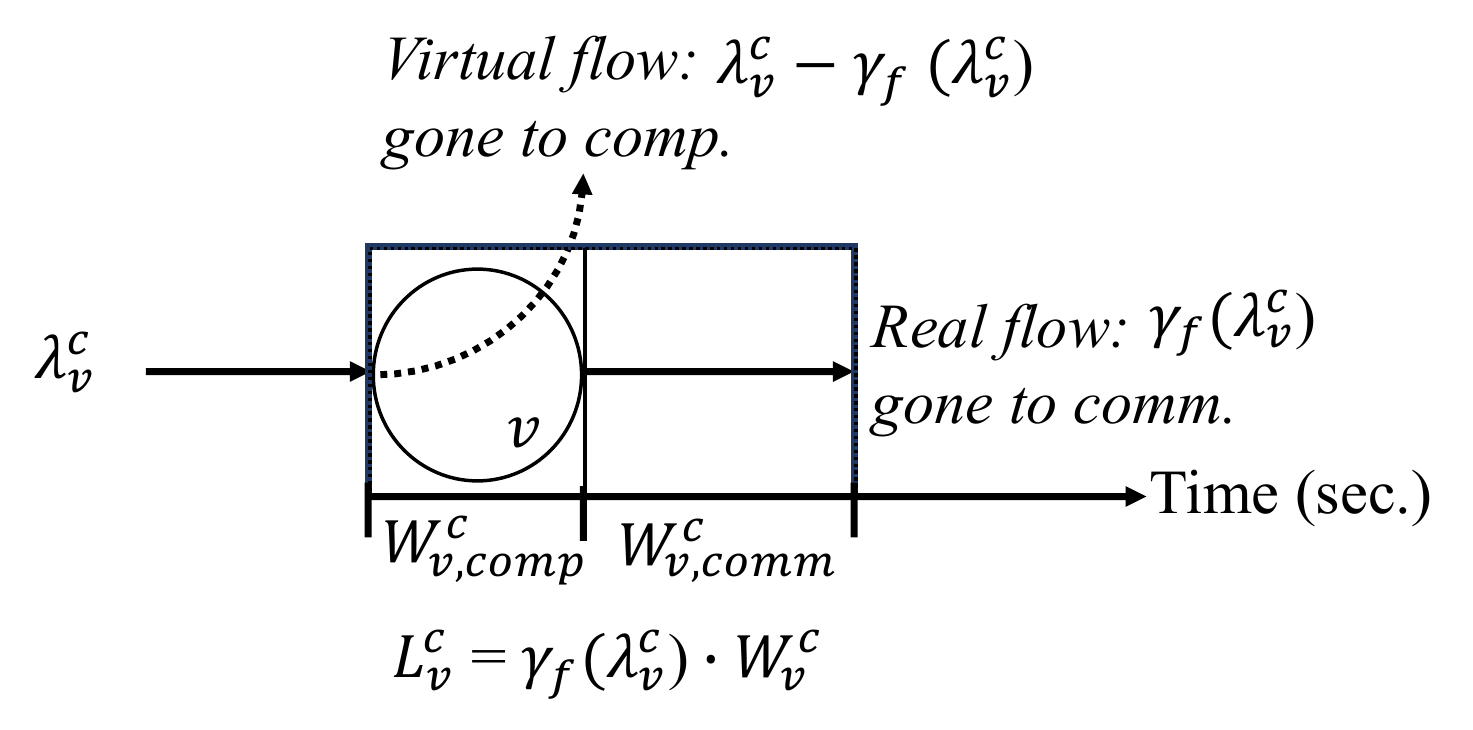}
\caption{Little's law in a computation scenario.}
\label{computeflow}
\end{figure} 

{\bf \em Computing with Little's Law.} We next introduce a novel relation for computation by building a connection between the limits of functional compression rates and Little's law in queueing theory. Little's law states that the long-term average number $L$ of packets in a stationary system is equal to the long-term average effective arrival rate $\lambda$ multiplied by the average time $W$ that a packet spends in the system. More formally, it can be expressed as $L=\lambda W$. The result applies to any system that is stable and non-preemptive, and the relationship does not depend on the distribution of the arrival process, the service distribution, and the service order \cite{Klein1975}.

By Little's Law, the long-term average number $L_v^c$ of packets in node $v$ for class $c\in C$ functions satisfies
\begin{align}
\label{LittleComputation}
L_v^c=\gamma_f(\lambda_v^c) \cdot W_{v}^c,
\end{align}
where $W_{v}^c$ is delay cost function that has 2 components: 
\begin{align}
\label{additive_delay_cost}
W_{v}^c=\ccomp+\ccomm 
\end{align}
that are positive and non-decreasing in the corresponding flows. We aim to infer the value of the computation flow rate generated for the input flow of rate $\lambda_v^c$ for class $c\in  C$ functions at $v$, i.e., $\gamma_f(\lambda_v^c)$, using Little's Law and the connection between $W_{v}^c$ and $L_v^c$ for the given network. It is also possible to generalize this to pipelining where the compute and communications tasks progress together. In that case the delay cost function satisfies $W_{v}^c=\max[\ccomp(m_v^c),\,\ccomm(n_v^c)]$.

Each node is equipped with a compute queue and communications queue which we will detail in Sect. \ref{costbreakdown}.  Therefore, $L_v^c=m_v^c+n_v^c$, where $m_v^c$ and $n_v^c$ are the  long-term  average  number  of  packets of class $c\in  C$ at node $v$ waiting for computation service and communications service, respectively, and are connected via (\ref{LittleComputation}). 
More precisely, the relations among $L_v^c$, $m_v^c$, and $n_v^c$ is determined as function of $\gamma_f(\lambda_v^c)$ as follows:
\begin{align}
\label{computationno}
m_v^c=L_v^c \cdot \left(1-\frac{\gamma_f(\lambda_v^c)}{\lambda_v^c}\right), \quad
n_v^c=L_v^c \cdot \frac{\gamma_f(\lambda_v^c)}{\lambda_v^c}.
\end{align}


We next introduce a novel flow-based notion for computing.

{\bf \em Characterizing the flow of computation.} 
The best functional compression rate attainable through a network can be given via extending the notion of graph entropy to capture the network's structure as well as function's surjectivity. 
We denote this rate by $\Hgf$ that satisfies $\Hgf\leq \gamma_f(\lambda_v^c)\leq \lambda_v^c < \mu_v^c$ where the lower bound is due to the limits of functional compression. 
To the best of our knowledge, there do not exist algorithms that exploit the theoretical limits provided by coloring characteristic graphs of functions to optimize $\gamma_f(\lambda_v^c)$. To capture the fundamental limits of compression, we next introduce a novel notion.
	
\begin{defi}{\bf Entropic surjectivity, $\Gamma_c$.} 
Entropic surjectivity of a function is how well the function $f_c: \X\to Y$, $c\in C$ can be compressed with respect to the compression rate of source symbols $\X$. We denote the entropic surjectivity of function $f_c$ with respect to $\X$ by
\begin{align}
\Gamma_c=\frac{\Hgf}{H(\X)}.
\end{align}
\end{defi}	
This is a measure of how well a network can compress a function that the destination wants to compute. Since surjective functions have high entropy via compression through a network, a function with high entropy yields a high entropic surjectivity, i.e., is harder to compress.

{\bf \em Flow requirements.} 
Note that $\Gamma_c$ is maximized when the function $f$ with domain $\X$ and codomain $Y$ is surjective, i.e., for every $y\in Y$ there exists at least one $\x\in \X$ with $f(\x)=y$. It is lower bounded by zero which is when the function maps all elements of $\X$ to the same element of $Y$. Therefore, $\Gamma$ can be used a measure of how surjective the function $f$ is.

Consider a function associated with class $c$, i.e., $f_c$. 
Total incoming flow rate needed is approximated as $H(\X)$. To compute $f_c(\X)$, the source needs to transmit a minimum of $\Hgf$ bits/source symbols. In this case, the proportion of flow that requires communications -- 
the proportion of generated flow as a result of computing 
-- needs to satisfy 
\begin{align}
\label{flow_entropy_relation}
\Gamma_c\leq \frac{\gamma_f(\lambda_v^c)}{\lambda_v^c}.
\end{align}	
Maximum reduction in communications flow, i.e., the flow that vanishes due to computation, is when the above is satisfied with equality, i.e., $\lambda_v^c-\gamma_f(\lambda_v^c)= \lambda_v^c (1-\Gamma_c)$ from the flow conservation constraints. 
In the regime of low compression, the network is communication intensive, i.e., $\gamma_f(\lambda_v^c)$ is high versus the computation intensive regime where higher compression is possible, i.e., $\gamma_f(\lambda_v^c)$ is low. In any regime, $\gamma_f(\lambda_v^c)$ has to be sufficiently large to ensure that the compute task is performed. For example, {\em Search function} has time complexity $\log(N)$ for an input size $N$ with rate $\lambda_v^c$. The entropy rate of the source satisfies $\Hg(\X)=H(X)=\lambda_v^c$ and the entropy rate of the function is $\Hgf\approx \log(\lambda_v^c)$. The {\em Classification function} has time complexity $\exp(N)$ and $\Hgf\approx \exp(\lambda_v^c)$.

We aim to approximate $\gamma_f(\lambda_v^c)$ using the connection between Little's law that relates the number of packets $L_v^c$ to function's entropic surjectivity $\Gamma_c$. To that end, we will next characterize the average computing time of different function classes for the proposed network setting. Later in Sect. \ref{flowanalysis} we will characterize the thresholds on $\gamma_f(\lambda_v^c)$ for successfully performing computations of different classes of functions.

\section{Cost Models}
\label{costbreakdown}
We explore the breakdown for the compute and communications time of a task and the 
time it spends in the system.

\subsection{Task Completion Time}
\label{taskcompletiontime}

{\bf \em Functional compression and surjection factor.} Intermediate computations at each node $v$ provides a compressed or refined representation. 
The cost of communications which equals the average waiting time, i.e., the sum of the average queueing and service times of a packet in $\s$, and given by the positive and convex function of the departing flow 
\begin{align}
\label{cost_comms}
\ccomm(n_v^c)=
\begin{cases}
\frac{1}{\mu_v^c-\gamma_f(\lambda_v^c)},\,\, v\in V',\, c\in  C,\\
\frac{1}{\mu_v^c-\beta^c},\,\, v=s,\, c\in  C, \\
0,\,\, v=0,\, c\in  C
\end{cases}
\end{align}	
is increasing in $n_v^c=\frac{\gamma_f(\lambda_v^c)}{\mu_v^c-\gamma_f(\lambda_v^c)}$ and upper bounded by $\frac{1}{\mu_v^c(1-\rho_v^c)}
$ if $\gamma_f(\lambda_v^c)=\lambda_v^c$ and $n_v^c=\frac{\rho_v^c}{1-\rho_v^c}$.  
It could instead be chosen to model the queue size, utilization, or the queuing probability. Our model does not account for the physical data transmission time, excluding the link layer aspects.

 
{\bf \em Average-case complexity.} We assume an average-case complexity model for computation where the algorithm does the computation in a sorted array (e.g., binary search).
Let $\df$ denote the time complexity for computing functions of class $c\in C$ packet at node $v\in V$ in the units of packets or bits.  
While $\df$ increases in the function's complexity, its behavior is determined by the function class $c$. 
It is given as function of the input size $m_v^c$, i.e., the number of packets needed to represent the input. 
Hence, $m_v^c$ denotes the number of packets of class $c$ at $v$ waiting for computation service. Assume that a class $c$ packet has an arrival rate $\lambda_v^c$ in $\mbox{bits}/\s$. Hence, the cost of computation is the average running time of packets for realizing the computation task, i.e., processing compute packets, in $\s$ 
\begin{align}
\label{cost_compute}
\ccomp(m_v^c) = 
\begin{cases}
\frac{1}{\lambda_v^c} \cdot \dfLvc, v\in V',\, c\in  C,\\
0,\,\, v=\{0,s\},\, c\in  C,
\end{cases}
\end{align} 
where we approximate time complexity of 
a flow of class $c$ at node $v$ as $\ccomp(m_v^c)\approx \frac{1}{\lambda_v^c} \cdot \df$, $v\in V',\, c\in  C$. The approximation is due to flow conservation we will detail later. 

Note that from (\ref{computationno})  $m_v^c$ increases in $\lambda_v^c$ because the processing factor increases always at a smaller rate than the incoming flow rate does. This in turn increases $\df$ and $\ccomp(m_v^c)$. 
Furthermore, the long-term average $L_v^c$ in (\ref{LittleComputation}) increases in $\lambda_v^c$ (under fixed $\mu_v^c$).  
Hence, the behavior of $\ccomm(n_v^c)$ is determined by (\ref{LittleComputation}). As a result $n_v^c$ increases if the scaling of $\gamma_f(\lambda_v^c)$ in $\lambda_v^c$ is linear, versus $n_v^c$ is less sensitive when the scaling is sublinear.

{\bf \em A taxonomy of functions.}
We consider three function categories and with different time complexities. For {\em Search function} that tries to locate an element in a sorted array, an algorithm runs in logarithmic time, which has low complexity. 
For {\em MapReduce function}, since the reduce functions of interest are linear, the algorithm runs in linear time, which is of medium complexity. For {\em Classification function}, we consider the set of all decision problems that have exponential runtime, which is of high complexity. 
The time complexity, i.e., the order of the count of operations, of these functions satisfies:
\begin{align}
\label{time_complexity_special_functions}
\df=\begin{cases}
O(\log(m_v^c)),\quad&\text{Search},\\
O(m_v^c),\quad&\text{MapReduce},\\
O(\exp(m_v^c)),\quad&\text{Classification}.\\
\end{cases}
\end{align}

Using the order of the count of operations in (\ref{time_complexity_special_functions}) and $\ccomp(m_v^c)$ in (\ref{cost_compute}), we can model the delay cost functions for computations of different classes of functions for $v\in V'$. 
In Sect. \ref{performance}, we will numerically investigate the behavior of computing cost. Note that when $\gamma_f(\lambda_v^c)=\lambda_v^c$ we have $m_v^c=0, \,\, n_v^c=\frac{\rho_v^c}{1-\rho_v^c}$ and $\ccomp(0)=0$, $\ccomm(n_v^c)= \frac{1}{\mu_v^c-\lambda_v^c}$, i.e., the computation time of the identity function is null.

In the special case when the computation cost is similar as the communication cost model, we assume that an arriving packet requires a compute service which is drawn from an exponential distribution with mean $1/\chi_v^c$ independently. Hence, the traffic intensity associated with class $c$ at node $v$ is $\sigma_v^c=(\lambda_v^c-\gamma_f(\lambda_v^c))/\chi_v^c\approx \lambda_v^c/\chi_v^c$. Hence, $\sigma_v=\sum\nolimits_{c\in C}\sigma_v^c$. Under the given assumptions and $m_v^c=\sigma_v^c/(1-\sigma_v^c)$, the compute cost for $v\in V'$ is given as
\begin{align}
\label{cost_compute}
\ccomp(m_v^c)=\frac{k_v^c}{\chi_v^c(1-\sigma_v^c)}=\frac{k_v^c m_v^c}{\chi_v^c \sigma_v^c}\,\, \s,
\end{align} 
which generalizes the cost of MapReduce in (\ref{time_complexity_special_functions}) due to the dependence of $\sigma_v^c$ on $m_v^c$. 
This model is yet to account for the routing information that we detail next. To generalize the above special case of convex compute cost in (\ref{cost_compute}) we will focus on a departure-based model (see (\ref{cost_compute_routing})).

\subsection{Task Load Balancing}
\label{loadbalancing}

In this section, we detail the routing model for different function classes and their conversions.

{\bf \em Causality and function class conversions.} 
In our  setting, at each node, computation is followed by computation. The average time a packet spends is given by the sum of the average times required by computation that is followed by communications, while satisfying the stability conditions. 
In general packets can change their classes when routed from one node to another \cite{nelson2013probability}. 
Hence, at each node, conversion between classes is allowed. Let $ C_v^{in}\subseteq C$ and $ C_v^{out}$ be the set of incoming and outgoing classes at node $v\in V$. We assume that $ C_v^{out}\subseteq C_v^{in}\subseteq C$ as computations on $ C_v^{in}$ can only reduce the number of function classes in the outgoing link. 
Each computation transforms a class $c\in C$ of flow into another class $c'\in C$ of flow. The transformations among classes will be characterized by a Markov chain, where the rates of transitions among the classes are known \cite[Ch. 10.3.8]{nelson2013probability}. We detail a Markov routing policy later in this section. Because $ C_v^{out}\subseteq C_v^{in}$, a computation transforms a flow class $c_i$ into a flow class $c_j$ if $j\geq i$. However, vice versa is not possible.

To refine the compute cost in (\ref{cost_compute}) for the networked setting, we allow packets departing from the network after service completion. Let $p_v^{dep}(c)=p_{v,0}^{rou}(c)$ denote the probability that a class $c$ packet departs upon completing its service at $v$, where $0\in V$ represents a virtual (shared destination) node that denotes the completion of a task. Hence, the departure rate of class $c$ packets from $v$ is $\gamma_f(\lambda_v^c) p_{v,0}^{rou}(c)$. We let $p_{v,w}^{rou}(c,c')$ be the probability that a class $c$ packet that finishes service at queue $v$ is routed to queue $w$ as a class $c'$ packet, and let $p_{v}^{rou}(c)=[p_{v,w}^{rou}(c,c')]_{c'\in C,\,w\in V}$. We will detail the concept of routing later in this section. Hence, we can refine the positive and non-decreasing compute cost in (\ref{cost_compute}) as 
\begin{align}
\label{cost_compute_routing}
\hspace{-0.25cm}\ccomp(m_v^c)\!=\!\frac{p_{v,0}^{rou}(c) m_v^c}{\lambda_v^c-\gamma_f(\lambda_v^c)}
\!+\!\sum\limits_{c'\in C}\sum\limits_{w\in V}\frac{p_{v,w}^{rou}(c,c') m_v^c}{\lambda_v^c-\gamma_f(\lambda_v^c) },
\end{align}
where the first term in the RHS captures the wasted computation rate due to departing packets of class $c$ in queue $v$ from the network, and the second term represents the cost of additional processing due to the routing of the packets of class $c\in C$ to other queues after finishing service from queue $v$ and transforming to any other class $c'\in C$. It also holds for every packet of a class $c$ that $p_{v,0}^{rou}(c)+\sum\limits_{c'\in C}\sum\limits_{w\in V}p_{v,w}^{rou}(c,c')=1$, i.e., upon leaving the compute queue of $v$, the packet is either routed to another node (while it might as well change its class) if its service is not completed, or departs the network.

{\bf \em Arrival/Departure/Function class conversion rates.}
Let the original arrival rate of class $c$ packets to the network be Poisson with rate $\beta^c$. Let $p_v^{arr}(c)=p_{s,v}^{rou}(c)$ be the probability that an arriving class $c$ packet is routed to queue $v$, where $s$ represents a virtual (shared) source node that denotes the origin of a task. Assuming that all arriving packets are assigned to a queue, we have that $\sum\limits_{v\in V}p_{s,v}^{rou}(c)=1$.  
We also assume that there is a virtual destination node $0\in V$ that indicates the completion of a task. 
The total departure rate of class $c$ packets from $v$ in the \textit{forward process} is given by $\lambda_v^c \cdot p_{v,0}^{rou}(c)$. With the above definitions, the routing probabilities between nodes and function classes have the structure:
\begin{align}
p_{s}^{rou}(c)&=[p_{s,v}^{rou}(c)]_{v \in V\backslash\{s\}}\in [0,1]^{(|V|-1)\times 1}\\
&=\left(\begin{smallmatrix} p_{s,v_1}^{rou}(c) & p_{s,v_2}^{rou}(c) & \dots & p_{s,v_{|V|-2}}^{rou}(c) & p_{s,0}^{rou}(c) \end{smallmatrix}\right)^{\intercal}.\nonumber
\end{align}
For $v\in V'$ we have the following matrix:
\begin{align}
\label{Proutingmulticlass}
p_{v}^{rou}(c)&=[p_{v,w}^{rou}(c,c')]_{w\in V\backslash\{s\},\,c'\in C}\in [0,1]^{(|V|-1)\times |C|}\\
=&\left(\begin{smallmatrix}
0 & \dots & 0 & p_{v,v_1}^{rou}(c,c) & p_{v,v_1}^{rou}(c,c+1) & \dots  & p_{v,v_1}^{rou}(c,|C|) \\
0 & \dots & 0 & 0 & p_{v,v_2}^{rou}(c+1,c+1) & \dots  & p_{v,v_2}^{rou}(c+1,|C|)\\
0 & \vdots & \vdots & \vdots & \vdots & \ddots & \vdots \\
0 & \dots & 0 & 0 & 0 & \dots  & p_{v,{v_{|V|-2}}}^{rou}(|C|,|C|) \\
0 & \dots & 0 & p_{v,0}^{rou}(c) & 0 & \dots & 0
\end{smallmatrix}\right)\nonumber
\end{align}
which is an upper-triangular matrix, where we note that $ C_v^{out}\subseteq C_v^{in}\subseteq C$, and conversion from $c$ to $c-1$ is not possible because $c-1$ has a higher complexity and this violates the purpose of computation, justifying the upper-triangular structure except for the sink nodes.

{\bf \em Routing for computing.} 
The network accommodates different types of state transitions, i.e., arrival, departure, routing, and internal state transitions \cite[10.6.3]{nelson2013probability}. We assume a Markov routing policy \cite[Ch. 10.6.2]{nelson2013probability} which can be described as follows. As a result of function computation, packets might have different classes and can change their class when routed from one node to another where routing probabilities depend on a packet's class. Let $p_{v,w}^{rou}(c,c')$ be the probability that a class $c$ packet that finishes service at node $v$ is routed to the compute queue of node $w$ as a class $c'$ packet. The probability that a class $c$ packet departs from the network after service completion at queue/node $v\in V'$ is given by
\begin{align}
p_{v,0}^{rou}(c)=1-\sum\limits_{w\in V'}\sum\limits_{c'\in C}p_{v,w}^{rou}(c,c'),
\end{align}
where the second term on the RHS denotes the total probability that the packets stay in the network. Since it is an open network model, for every class $c$ there is at least one value of $v$ so that $p_{v,0}^{rou}(c) > 0$. Thus all packets eventually leave the system. 

{\bf \em Flow conservation principles.}  
A node $v\in V$ can forward the flows it receives through the incoming edges of $v$, and generate a flow of class $c$ which is modeled via a self-loop by consuming/terminating equal amounts of incoming flows of class $c'\in C$. 
Conversion among classes is also possible. In the stationary regime the total arrival rate of class $c\in C$ packets to node $v$ is the lumped sum of the flow of class $c$ packets through its incoming edges or the flow via the self-loop. The flow that is gone to computation is the difference between the incoming flow and the generated flow, i.e., $\lambda_v^c-\gamma_f(\lambda_v^c)$. Due to flow conservation for each class, we have
\begin{align}
\label{throughputofclassc}
\hspace{-0.4cm}\lambda_v^c= 
\begin{cases}
\beta^c,\, v=s,\\ 
\beta_v^c + \sum\limits_{w\in V'}\sum\limits_{c'\in C}\gamma_f(\lambda_{w}^{c'}) p_{w,v}^{rou}(c',c), \, v\in V',\\
\sum\limits_{v\in V'} \gamma_f(\lambda_v^c) p_{v,0}^{rou}(c),\, v=0,
\end{cases}
\end{align}
where $\beta_v^c=\beta^c \cdot p_{s,v}^{rou}(c)$ denotes the original arrival rate of class $c$ packets that are assigned to the compute queue of node $v$, and the second term on the RHS denotes the aggregate arrival rate of packets that are routed to queue $v$ as a class $c$ packet after finishing service at other queues $w\in V$ as a class $c'\in C$. Note also that the term $\gamma_f(\lambda_{w}^{c'})$ denotes the total departure rate of class $c'$ packets from node $w$ rate of class $c'$ packets to queue $w$ (as a result of computation).


\begin{figure}[t!]
\centering
\includegraphics[width=0.8\columnwidth]{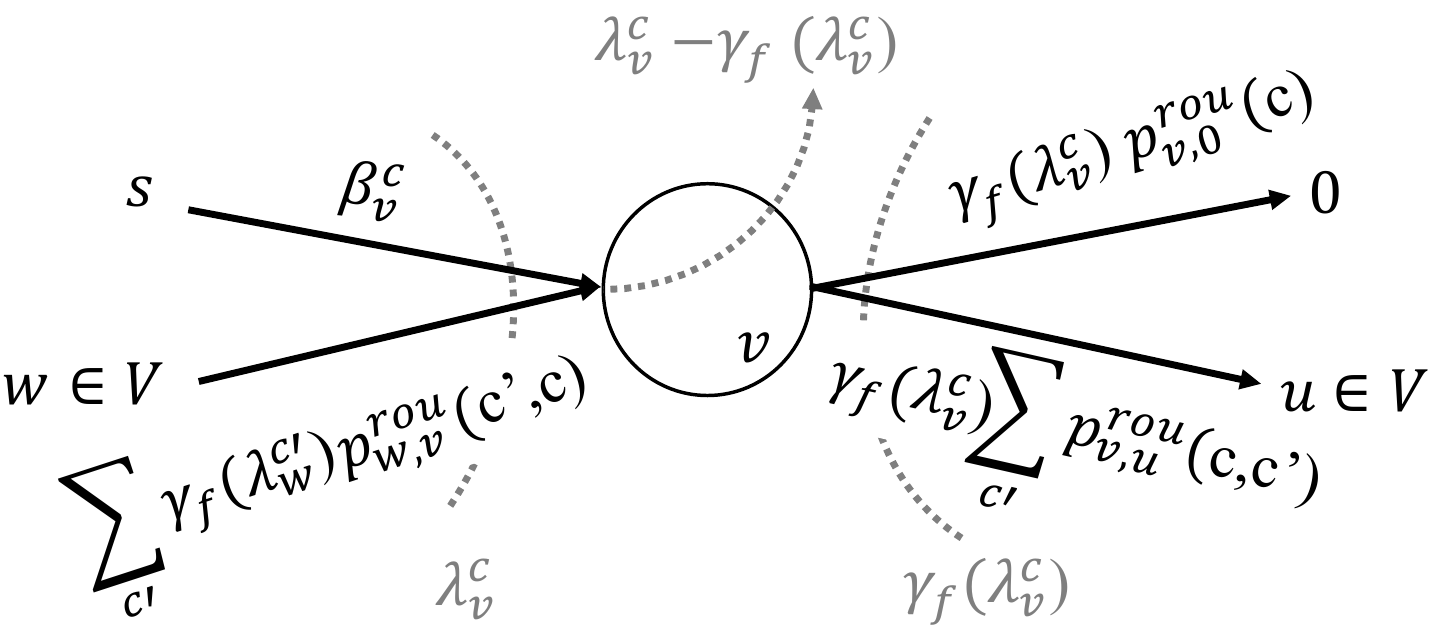}
\caption{Computational flow breakdown at $v\in V'$ where min cut $\lambda_v^c$ is the total arrival rate of flow $c$ 
that incorporates the original arrivals $\beta_v^c$ and the arrivals routed from $w \in V'$. Min cut $\gamma_f(\lambda_v^c)$ is the total generated rate of $c$, 
and the departures are routed to $u\in V'$.}
\label{flowpernode}
\end{figure} 		
In Fig. \ref{flowpernode}, we consider a node in isolation to illustrate the computational and communication flows at a typical node $v\in V'$ of the Jackson network. The min-cut that denotes the total arrival rate of computational flow $c$ is $\lambda_v^c$. This via (\ref{throughputofclassc}) captures the rate of original arrivals which is $\beta_v^c$, and the arrivals routed from any other node $w\in V'$ in the network. If there is no $w$ such that $p_{w,v}^{rou}(c',c)>0$, then $\lambda_v^c=\beta_v^c$. The cut $\gamma_f(\lambda_v^c)$ denotes the total generated rate (or processing factor) of computational flow $c$ at node $v$. The processed flow can be routed to any $u\in V'$ in the network if $p_{v,u}^{rou}(c,c')>0$. If there is no such node, then $\gamma_f(\lambda_v^c)$ departs the system.

{\bf \em Higher order properties.} The communication network $\mathcal{N}$ under the Markov routing policy can be modeled using a discrete-time Markov chain (DTMC) 
$Y_1, Y_2, Y_3, \dots$. The entropy rate for 
$(Y_k)$ on a countable number of states is $ H(Y) = - \sum_{ij} \pi_i P_{ij} \log P_{ij}$ where $\pi_i$ is the limiting distribution 
and $(P_{ij})$ 
is the state transition matrix. 
The rate $H(Y)$ can capture the routing information and characterize the distributions, i.e., the higher order properties, of 
the costs $\ccomm$ and $\ccomp$ incurred being in each 
state of the network. 
It can also explain the relation between routings and connectivity and how long it will take to compute. To that end, we next associate the states of the DTMC to a slotted computation model.

{\bf \em Characterizing the computing cost on a time-slot basis.} At each time slot $t$, the Markov chain is in state $i$ with probability $\pi_i$. We infer the following quantities at $t$, each indicated by $^{(t)}$ as function of the state $i$:
\begin{CompactEnumerate}
    \item ${n_v^c}^{(t)}$ and $\gamma_f({\lambda_v^c}^{(t)})$, and $\ccommt({n_v^c}^{(t)})$ using (\ref{cost_comms}),
    \item ${m_v^c}^{(t)}={n_v^c}^{(t)}\big(\frac{{\lambda_v^c}^{(t)}}{\gamma_f({\lambda_v^c}^{(t)})}-1\big)$ using (\ref{computationno}),
    \item $\ccompt({m_v^c}^{(t)})$ using (\ref{cost_compute}), and
    \item ${W_v^c}^{(t)}=\ccompt+\ccommt$, and its relation to $\gamma_f({\lambda_v^c}^{(t)})$.
\end{CompactEnumerate}

A class of function is easy to compute if the generated rate per unit time of computation is high. More precisely, in a DTMC with fixed slot duration $T$, the generated flow rate across $t_{\max}$ slots is $T\sum_{t=1}^{t_{\max}}{\gamma_f^{(t)}({\lambda_v^c}^{(t)})}$ in bits per second. If this number is large, then $f_c$ is easier to compute for given routings or connectivity. Hence, using this approach we can infer what classes of functions can be computed easily. 
If on the other hand, to account for the cost of computing more precisely, we assume a continuous-time Markov chain, we can instead consider a DTMC $Y_k$ to describe the $k^{th}$ jump of the process where the variables $S_1,S_2,S_3,\dots$ describe holding times in each state. 
In state $i$, the slot duration ${W_v^c}^{(t)}$ (holding time) will be approximated by the average time needed to perform computation on the distribution of classes, ${n_v^c}^{(t)}$, $c\in C$ in the current state, i.e., $\mathbb{E}[W_i]$. 
Hence, ${W_v^c}^{(t)}$ is the realization of the holding time sampled from the exponential distribution with rate parameter $-1/\mathbb{E}[W_i]$. Then the total number of bits generated over $t_{\max}$ slots is given by $\sum_{t=1}^{t_{\max}}{W_v^c}^{(t)}\cdot{\gamma_f^{(t)}({\lambda_v^c}^{(t)})}$ and the total cost is $\sum_{t=1}^{t_{\max}}{W_v^c}^{(t)}$. 

{\bf \em Routing information $H(Y)$ as a proxy for communication cost.} Given routings $P_{ij}$, the limiting distribution $\pi_i$ determines the network state $[\mathbf{n}; \mathbf{m}]$, which provides insights into computing different classes of functions. To that end we seek the relation between $H(Y)$ and $\Hgf=\Gamma_c H(\X)$ where 
$H(\X)$ 
is known.  
The entropy rate of $\{\gamma_f(\lambda_v^c)\}_{v\in V,c\in C}$ is the same as 
$H(Y)$ because the DTMC $(Y_k)$ fully describes $\gamma_f({\lambda_v^c})$. Thus, $H(Y)$ determines how balanced the flows for the different classes are. 
If the stochastic process $(Y_k)$ is i.i.d., $H(Y)=H(Y_k)$, $k=1,\dots, N$, the class distribution across different 
states is not distinguished, i.e., $\gamma_f(\lambda_v^c)$ is unchanged in $c\in C$. 
On the other hand, $H(Y)$ is lowered by mixing tasks $f_c$ of different complexities (due to concavity of entropy, i.e., $H(Y) 
\leq -\sum_j (\sum_i \pi_i P_{ij})\log(\sum_i \pi_i P_{ij})$ where equality is when the process is i.i.d.). In this case the DTMC states are not visited with the same frequency, i.e., $\pi_i$ is not uniform, as some tasks demand high resources $n_v^c$ and have higher $\gamma_f(\lambda_v^c)$. This may help reduce the use of communication resources. 
In this case, mixing function classes can provide resource savings in networked environments and help lower $H(Y)$. Hence, in a 
Jackson network with nodes having similar routing probabilities and each state being visited as often, it might not be possible to effectively compute different function classes  
versus 
in a more structured network with transition probabilities being tailored for tasks, i.e., task-based or weighted link reservations, to enable mixing of 
classes.

In Table \ref{notationtable} we provide some notation we use in the paper.

\begin{table}[h!]\footnotesize
\centering
\begin{tabular}{| l | l |}
\hline
{\bf Definition} & {\bf Function} \\
\hline
Set of function classes; Function of class $c$ & $ C=\{c\}$; $f_c(\X)$\\
\hline
Entropy of source & $H(\X)$\\
\hline
Entropic surjectivity; Graph entropy  of $f_c(\X)$ & $\Gamma_c$; $\Hgf$\\
\hline
Total arrival rate; Service rate; Traffic intensity & $\lambda_v^c$; $\mu_v^c$; $\rho_v^c$\\
Long-term average number of  packets   & $L_v^c$\\
Breakdown of $L_v^c$ into comms and comp. & $n_v^c,\,m_v^c$\\
Cost of communications & $\ccomm(n_v^c)$\\
Cost of computation & $\ccomp(m_v^c)$\\
Original arrival rate of computational flow & $\beta_v^c$\\
Computation flow rate generated from $\lambda_v^c$ & $\gamma_f(\lambda_v^c)$\\
Compute 
rate; Traffic intensity of computing & $\chi_v^c$; $\sigma_v^c$\\
Time complexity of generating/processing a flow  & $\df$\\
\quad corresponding to class $c$ packets at node $v$ & \\
\hline
Probability of a class $c$ packet finishing service  & \\
\quad at node $v$ being routed to node $w$ as class $c'$ & $p_{v,w}^{rou}(c,c')$\\
Routing matrix of class $c$ packets from node $v$ & $p_v^{rou}(c)$\\
\hline
\end{tabular}
\caption{Notation.}
\label{notationtable}
\end{table}

\section{Computation Flow Analysis}
\label{flowanalysis}
In the previous section, we explored how given routings or connectivity can tell us regarding the class of functions we can compute. In this section we are interested in the reverse question of how to optimize the routings and how the connectivity looks like to compute a {\em class of functions}. To that end investigate where to compute and how to compute a function class $f_c$ of known time complexity to optimize the average computation time in the Jackson network.

{\bf \em Computation allocation.} 
Our goal is to understand whether a divide-and-conquer-based approach is more favorable than a centralized approach. In divide-and-conquer a subset of nodes work on the sub-problems of a given task, which are to be combined to compute the target value either at any node, including the destination. In a centralized approach tasks are not split into sub-problems. Instead they are run at once.

{\bf \em Processing of single flow.} We first consider the distribution and cost of single flow in isolation in the Jackson network and we drop the superscript $c$ due to the single class assumption.

\begin{ex}\label{ex_surjection_factor_singleflow_search}

{\bf \em Bisection -- tree branching over the network.} The network computes $f(\X)=\min \X$ using the bisection method. The arrivals $\X$ are uniformly split among nodes, i.e., $\beta_v=\beta$ $\forall v$. Node $v$ works on the set ${\X}_v=(X_{(v-1)\frac{N}{|V|}},\dots, X_{v\frac{N}{|V|}})$ and computes $f(\X_v)$ locally, i.e., $p_{s,v}^{rou}=\frac{|V|}{N}$. 

(i) No intermediate routing is allowed, i.e., $p_{v,0}^{rou}=1$. The computation outcome is directly routed to $0\in V$ that decides the final outcome $\min_{v\in V} f(\X_v)=f(\X)$. The computation time complexity of the initial stage per node is $O(\log(\frac{N}{|V|}))$ and the routing stage is $\log(|V|)$. Hence, the total computation complexity is $|V|O(\log(\frac{N}{|V|}))+\log(|V|)$. The communications cost is due to the routings $(v,0)$ for $v\in V'$, which in total gives $|V|\exp(1)$ (every node does 1 transmission).

(ii) Intermediate routings are allowed. We pick a subset of nodes $W\subset V$ with $|W|=O(\log(|V|))$ to run the bisection algorithm, i.e., $p_{v,w}^{rou}\approx \frac{1}{\log(|V|)}, \, w\in W$ and $p_{v,w}^{rou}=0,\, w\in V\backslash W$. Hence, the time complexity of the initial stage per $v\in V\backslash W$ is $O(\log(\frac{N}{|V|}))$. Node $w\in W$ then computes $f(\X_w)=\min\nolimits_{v: p_{v,w}^{rou}>0} f(\X_v)$, and routes intermediate computation to $0\in V$ which in turn computes $\min_{w\in W} f(\X_w)=f(\X)$ that has time complexity $O(\log(|W|))$. The total compute complexity is $(|V|-|W|)O(\log(\frac{N}{|V|}))+|W|\log(\frac{|V|-|W|}{|W|})+\log(|W|)$. The communications cost is due to the routings $(v,w)$ and $(w,0)$ for $v\in V',\, w\in W$, which in aggregate gives $(|V|-|W|)\exp(1)+|W|\exp(1)=|V|\exp(1)$. 

From (i)-(ii) the total cost is determined by the computation complexity, which can be reduced via intermediate routings. 
The highest gains are possible when the depth of the network grows like $O(\log(N))$ that favors divide-and-conquer. A similar approach will follow for MapReduce.

\end{ex}

We next consider another single flow example.

\begin{ex}\label{ex_surjection_factor_singleflow_classification}

{\bf \em Classification.} We consider the classification problem, i.e., the problem of identifying to which of a set of categories an observation $\X$ belongs. In {\em linear classification}, the predicted category is the one with the highest score, where the score function has the following dot product form $f(\X,l)=score(\X,l)={\bm \beta}_l\cdot \X = \sum_{k=1}^N \beta_{kl} X_k$ where ${\bm \beta}_l$ is the vector of weights corresponding to category $l$, and $f(\X,l)$ is the score associated with assigning $\X$ to $l$ (which represents the  utility associated with $\X$ being in category $l$).

(i) If data is not split, then the computation cost is the cost equivalent of $N$ multiplications, i.e., $O(N\log(N))$.

(ii) If data is split between a subset of nodes $W$ such that ${\X}_w=(X_{(w-1)\frac{N}{|W|}},\dots, X_{w\frac{N}{|W|}})$, then $\{X_k\}_k$ and $\{\beta_{kl}\}_l$ should also be provided for accurately computing the score. This requires a cost equivalent to $O(\frac{N}{|W|}\log N)$ where $\log N$ bits is to quantify the locations of $\{X_k\}_k$ for each $w$. The compute cost is $O(\frac{N}{|W|}\log N \log(\frac{N}{|W|}\log N))$. Coordinating such intermediate computations requires $|W|$ additions, yielding a total cost of  $O(\frac{N}{|W|}\log N \log(\frac{N}{|W|}\log N))+O(|W|)$. To ensure that the computation cost in (ii) is smaller than the one of (i), $|W|\geq O(\log N)$. Otherwise, the compute cost will be higher than $O(N\log N)$. The communications cost of forwarding the intermediate results of $w\in W$ is determined by the flow amount which is at least $H(f({\X}_w,l))$ at $w\in W$. If $|W|$ is small, score estimates are easy to obtain due to the weak law of large numbers 
(a large ratio $\frac{N}{|W|}$ gives a good estimate of scores $f({\X}_w,l)$ as the sample average converges in probability to the expected value, yielding a low entropy), but if $|W|$ is large, this might not be possible, causing a high $H(f({\X}_w,l))$ (e.g., if the classifier is sensitive to ${\bm \beta}_l$ or ${\bm \beta}_l$ need to be trained, incurring a high cost for task distribution) that in turn requires high communication cost per $w\in W$. As the communication costs accumulate, distributing such flows might not be favorable. 
However, if the weights ${\bm \beta}_l$ for category $l$ are not 
sensitive to the data coordinates $K=1,\dots, N$ of $\X$, i.e., the observations are easy to classify, then the communications overhead can be lowered). Hence, if splitting does not cause 
jumps in $H(f({\X}_w,l))$, it might be favorable. 
\end{ex}

From Examples \ref{ex_surjection_factor_singleflow_search}, \ref{ex_surjection_factor_singleflow_classification}, how to distribute tasks depends on the task complexity and the associated communication resources to leverage the distributed computation. Tasks with high complexity, as long as the compute resources are sufficient, require centralized processing. If one node is not that powerful, then distribute (since $\gamma_f(\lambda_v^c)$ will be large). 
It is also intuitive that low complexity tasks, such as Search and MapReduce can be distributed over the network by splitting dataset. However, tasks with high complexity such as Classification might not be feasible in a distributed manner because one needs the whole dataset (as we detailed in Example \ref{ex_surjection_factor_singleflow_classification}).
To understand how $\gamma_f(\lambda_v^c)$ impacts the processing, we next focus on the simultaneous processing of multiple flows.

{\bf \em Processing of multiple flows.} We explore how to allocate the computations  to conduct the computation task effectively without requiring the protocol information. To that end we next sketch the behavior of cost of 3 different function classes. 

\begin{ex}\label{ex_surjection_factor_multiflow}

{\bf \em Case I. Sublinear surjection factor.} 
We assume that $\gamma_f(\lambda_v)= o(\lambda_v^c)$. This implies from (\ref{computationno}) that $L_v\propto m_v^c$ and $n_v^c=o(L_v^c)=o(m_v^c)$. Due to the convexity of the communication cost, from (\ref{cost_comms}) $n_v^c=O(\exp(\gamma_f(\lambda_v^c)))=O(\lambda_v^c)$. In this case, the costs corresponding to different classes are
\begin{align}
\ccomp(m_v^c)=\begin{cases}
\frac{1}{\lambda_v^c} O(n_v^c),\quad&\text{Search},\\ 
\frac{1}{\lambda_v^c} O(\exp(n_v^c)),\quad&\text{MapReduce},\\
\frac{1}{\lambda_v^c} O(\exp(\exp(n_v^c))),\quad&\text{Classification}.
\end{cases}\nonumber
\end{align}

{\bf \em Case II. Linear/Proportional surjection factor.} 
We assume that $\gamma_f(\lambda_v)= O(\lambda_v^c)$, implying that $L_v=O(m_v^c)=O(n_v^c)$. From (\ref{cost_comms}) $n_v^c=O(\exp(\lambda_v^c))$. In this case, the costs corresponding to different classes can be written as
\begin{align}
\ccomp(m_v^c)=\begin{cases}
\frac{1}{\lambda_v^c} O(\log\big(n_v^c\big)),\quad&\text{Search},\\ 
\frac{1}{\lambda_v^c} O(n_v^c),\quad&\text{MapReduce},\\
\frac{1}{\lambda_v^c} O(\exp(n_v^c)),\quad&\text{Classification}.
\end{cases}\nonumber
\end{align}
In Case I, the computation cost is either very high because $(\lambda_v^c)^*$ is high (due to $\ccomp(m_v^c)\geq \ccomm(n_v^c)$) or relatively higher than the communication cost. In these regimes, low complexity tasks, such as Search and MapReduce should be done distributedly (because the dynamic range of $\lambda_v^c$ is smaller), and high complexity tasks can be run centralized provided that compute resources suffice. On the other hand, in Case II, the communication cost dominates the computation cost as $(\lambda_v^c)^*$ is low, leading to a more distributed setting.

\end{ex}

{We next contrast the costs for Cases I-II. Task distribution 
suits for noncomplex tasks Search and MapReduce (Case I, low $\gamma_f(\lambda_v^c)$) since 
the maximum $\lambda_v^c$ that $v$ supports, $(\lambda_v^c)^*$, grows with the complexity, allowing centralized processing of complex tasks, e.g., Classification, as shown in Fig. \ref{ComputationAllocation}. 
If $\ccomm$ dominates the cost (Case II, high $\gamma_f(\lambda_v^c)$), a divide-and-conquer-based approach is favorable as $(\lambda_v^c)^*$ is low. However, 
when $\ccomm$ is rather negligible, 
the dynamic range of $\lambda_v^c$ grows and centralized processing may be feasible. Hence, for low (high) complexity tasks the network may operate in a connected (isolated) fashion. If the task complexity is heterogeneous, task-based link reservations become favorable, and the network may be sparsely connected as $\ccomm$ starts to dominate. 
If the routings are symmetric, then the tasks are distributed in which case realizing communication intensive 
tasks in a distributed fashion may cause high $\ccomm$.}

{\bf \em Routing schemes.} Fig. \ref{ComputationAllocation} also provides insights into flow deviation between different tasks, where the application of successive flow deviations leads to local minima as in the gradient method \cite{fratta1973flow}.  
For example, for given processing capability $\mu_v$ at $v\in V$  flows of high complexity tasks requiring centralized processing (due to high $(\lambda_v^c)^*$) versus low complexity tasks running in a distributed manner (due to small $(\lambda_v^c)^*$) can be traded-off. 
We observe that Case I 
is tailored for a hybrid regime, i.e., a mixture of distributed and centralized operation of tasks which can be done via flow deviation between tasks, and Case II 
enforces distributed computing of the tasks. 

{\bf \em Multiple flows and class conversions.} For simplicity of the exposition, we next consider a simple topology with $|V|=3$ nodes with identical processing capabilities, and distribute the different function classes (Search, MapReduce, Classification) 
on them. 

(i) Separating flows. Each node works on a distinct function class independently. Hence, the compute cost is separable and independent from the routings. The communications cost is only between links $(v,0)$, $v\in V'$ and can be computed using (\ref{cost_comms}) where $\gamma_f(\lambda_v^c)=\Hgf$. Using (\ref{computationno}), $m_v^c$ and $n_v^c$ can be determined as function of $L_v^c$ and $W_v^c$ is nonzero at each node only for the function class it works on and it satisfies (\ref{additive_delay_cost}). Hence, total cost is easily determined. 

(ii) Mixing flows. Each class is routed based on $p_{s}^{rou}(c)=[\frac{1}{3},\frac{1}{3},\frac{1}{3}]$, computation is evenly split among the nodes, and  $p_{v}^{rou}(c)=[p_{v,v_1}^{rou}(c),p_{v,v_2}^{rou}(c),p_{v,v_3}^{rou}(c),p_{v,0}^{rou}(c)]^{\intercal}$ where $p_{v,0}^{rou}(c)=\frac{\Hgf}{H(\X)}$ and $p_{v,w}^{rou}(c)=\frac{1-p_{v,0}^{rou}(c)}{3}$ for $w,v\in V'$.

We will simulate this setting in Sect. \ref{performance} and show that under some conditions on the surjection factor $\gamma_f(\lambda_v^c)$ and processing power $\mu_v^c$ it is possible to support a higher $L_v^c$ via mixing flows.

\begin{figure}[t!]
\centering
\includegraphics[width=0.9\columnwidth]{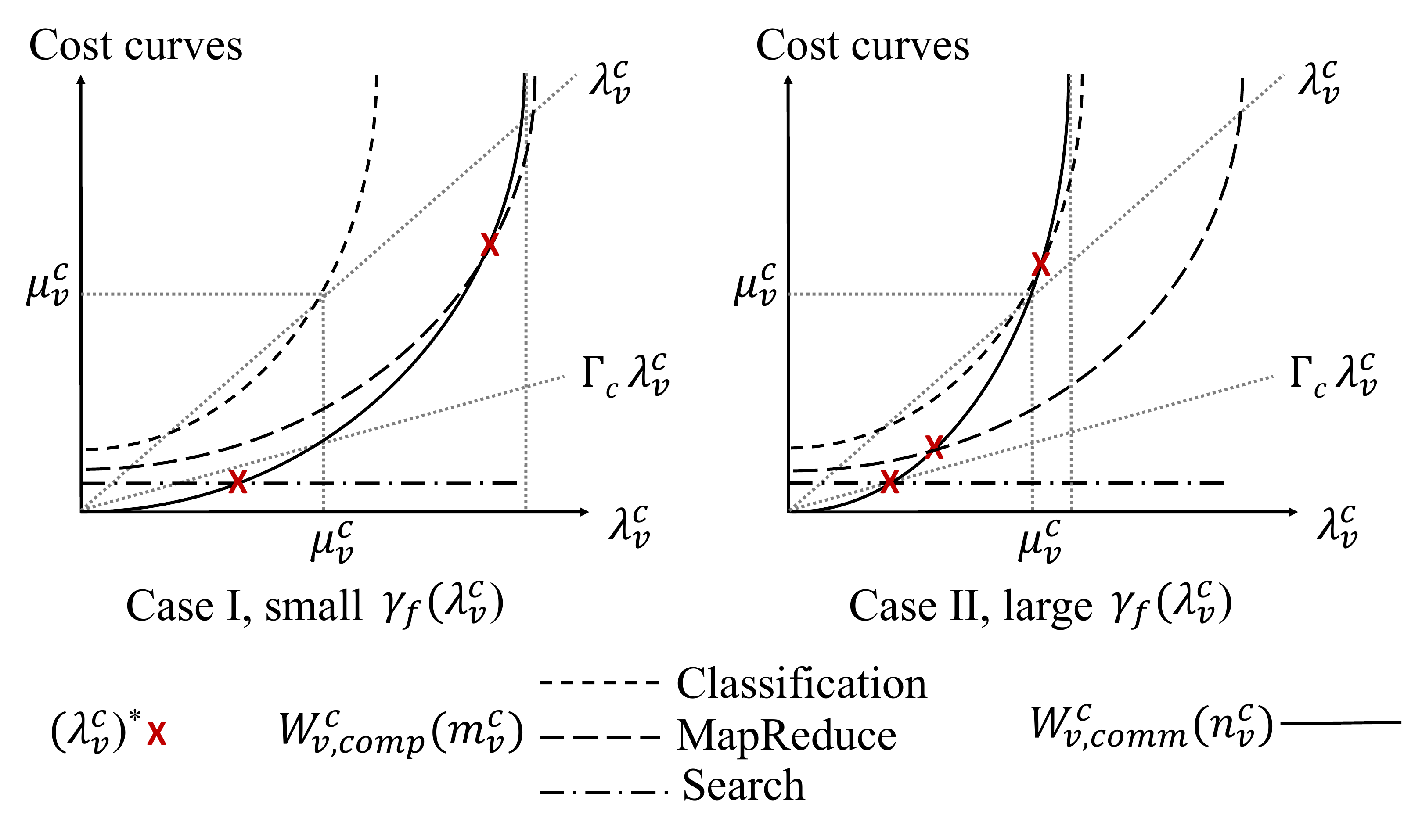}
\caption{Cost versus $\lambda_v^c$ where different line types represent distinct functions and the solid convex curve corresponds to $\ccomm$. (Left) Sublinear surjection where $\lambda_v^c\in [0,(\lambda_v^c)^*]$ has a high dynamic range. (Right) Superlinear surjection where $\ccomm$ is steep, leading to distributed computing of functions due to smaller range of $\lambda_v^c$.}
\label{ComputationAllocation}
\end{figure}

\subsection{Load Thresholds for Computing}
\label{loadthresholds}
In this section we explore the fundamental limits of the traffic intensities or loads associated with different computation tasks. We next define a threshold on traffic intensity $\rho_{th}$ such that for the range $\rho_v^c<\rho_{th}$ computation is allowed, 
and no computation is allowed for $\rho_v^c>\rho_{th}$.

\begin{prop}\label{LoadThreshold}{\bf A load threshold for distributed computing.}
A node $v\in V$ can do computation of a class $c\in C$ function if the following condition is satisfied:
\begin{align}
\label{condition_computation}
\rho_{th}=\min\limits_{\rho_v^c\geq 0}\Big[\rho_v^c \Big\vert \frac{(\rho_v^c)^2}{1-\rho_v^c}> \df\frac{1-\rho_v^c\Gamma_c}{1-\Gamma_c}\Big],
\end{align}
where $\rho_{th}\to 1$ as $m_v^c\to\infty$ since $\df$ 
increases in $m_v^c$.
\end{prop}

\begin{proof}
$\rho_{th}$ is obtained by comparing the delays with and without computation, on a per node basis. 
If the delay caused only by communication is higher than the total delay caused by computation followed by communication $W_v^c$ given by the lump sum of $\ccomp(m_v^c)$ in (\ref{cost_compute}) and $\ccomm(n_v^c)$ in (\ref{cost_comms}), i.e., the following condition is satisfied at $v\in V'$:
\begin{align}
\frac{1}{\mu_v^c(1-\rho_v^c)} > \frac{1}{\lambda_v^c}\df+\frac{1}{\mu_v^c\left(1-\rho_v^c\Gamma_c\right)},\nonumber
\end{align}
where $\gamma_f(\lambda_v^c)=\lambda_v^c\Gamma_c$, then $v$ decides to compute. 
\end{proof}

We next provide a necessary condition for stability of the 
network model for computation. 
\begin{prop}\label{stability}
The computation network is stable if the computation delay is at least as much as the communications delay, i.e., if the following condition is satisfied: 
$$\df\geq n_v^c,\quad c\in C,\,\, v\in V.$$ 
\end{prop}

\begin{proof}
For  stability,  the  long-term  average  number  of  packets  in the communications queue of $v$,  i.e., $n_v^c$,  should  be upper  bounded  by  the  long-term  average  number  of  packets in the compute queue of $v$,  i.e., $m_v^c$. 
Assume that $\frac{1}{\lambda_v^c}\df\geq \frac{1}{\mu_v^c}$. If this assumption did not hold, we would have $\df< \frac{\lambda_v^c}{\mu_v^c}
<\frac{\lambda_v^c}{\mu_v^c-\gamma_f(\lambda_v^c)}=n_v^c$, and if $\df<n_v^c$, then we have  equivalently the following relation would hold:
\begin{align}
\ccomp(m_v^c)=\frac{\df}{\lambda_v^c}<\ccomm(n_v^c)=\frac{1}{\mu_v^c-\gamma_f(\lambda_v^c)}.\nonumber
\end{align}
In this case, $n_v^c$ will accumulate over time, which will violate stability. Hence, delay  of  computation should  be  higher.
\end{proof}

\begin{prop}\label{rateofcomputationflow1}
{\bf An extension of Little's Law to computing.} The long-term average number $L_v^c$ of packets in node $v$ for class $c$ flow with time complexity $\df$ is bounded as
\begin{align}
L_v^c&\geq {b_v^c}^-\cdot\Big[\frac{\df}{\lambda_v^c}+\frac{1}{\mu_v^c-{b_v^c}^-}\Big],\nonumber\\
L_v^c&\leq {b_v^c}^+\cdot\Big[\frac{\df}{\lambda_v^c}+\frac{1}{\mu_v^c-{b_v^c}^+}\Big],\nonumber
\end{align}
where $(b_v^c)^{\pm}$ satisfies $2(b_v^c)^{\pm}=\lambda_v^c(1+\frac{1}{\df})+\mu_v^c\pm \sqrt{(\lambda_v^c(1+\frac{1}{\df})+\mu_v^c)^2-4\lambda_v^c\mu_v^c}$. 
\end{prop}

\begin{proof}
Following from Little's law in (\ref{LittleComputation}) we have
\begin{align}
 \df \leq L_v^c= \gamma_f(\lambda_v^c)\Big[\frac{\df}{\lambda_v^c}+\frac{1}{\mu_v^c-\gamma_f(\lambda_v^c)}\Big].\nonumber
\end{align}
Rearranging the above term for $2a_v^c=\lambda_v^c(1+1/\df)+\mu_v^c$ we obtain the following relation: $\gamma_f(\lambda_v^c)^2-2a_v^c\gamma_f(\lambda_v^c)+\lambda_v^c\mu_v^c\leq 0$ that gives the following range for $\gamma_f(\lambda_v^c)$:
\begin{align}
\label{gamma_range}
[a_v^c- \sqrt{(a_v^c)^2-\lambda_v^c\mu_v^c},\,a_v^c+ \sqrt{(a_v^c)^2-\lambda_v^c\mu_v^c}],
\end{align}
from which we observe that $\gamma_f(\lambda_v^c)$ is sublinear, i.e.,  $\gamma_f(\lambda_v^c)=o(\lambda_v^c)$ which is true if $(a_v^c)^2> \lambda_v^c\mu_v^c$, and is linear, i.e.,  $\gamma_f(\lambda_v^c)=O(\lambda_v^c)$ which is true if $(a_v^c)^2\approx\lambda_v^c\mu_v^c$. 
This is also intuitive from the surjection factors (Cases I-II) of Example \ref{ex_surjection_factor_multiflow}. 
From (\ref{LittleComputation}), (\ref{cost_comms}), and (\ref{cost_compute}), we get the desired result. 
Furthermore, if $\gamma_f(\lambda_v^c)\approx a_v^c$, the range provided in (\ref{gamma_range}) is tight and $L_v^c\approx  \sqrt{\frac{\mu_v^c}{\lambda_v^c}}\df+\frac{1}{\sqrt{\mu_v^c/\lambda_v^c}-1}\geq \sqrt{\frac{\lambda_v^c}{\mu_v^c}}+\frac{1}{\sqrt{\mu_v^c/\lambda_v^c}-1}$. 
Hence, the best achievable scaling is $L_v^c=O(\sqrt{\df})$.
\end{proof}

\begin{remark}
From Prop. \ref{rateofcomputationflow1}, we observe that as $\df$ increases, $a_v^c$ decreases and $\gamma_f(\lambda_v^c)$ concentrates. Ignoring this principle, if the processed flow rate were increased with the time complexity of the function, then the cost for both computation and communications would increase together. 
However, the processing factor can decrease with the time complexity, and the output rate may not be compressed below $\Hgf$. Hence,  
$\Hgf\leq \gamma_f(\lambda_v^c)\leq \Hg(\X)$ should be satisfied, where the upper limit is due to the identity function.
\end{remark}

We denote by $M_v^c$ the long-term average number of packets belonging to class $c\in C$ in $v\in V'$ waiting for communications service in case of no computation, i.e., $M_v^c=\frac{\lambda_v^c}{\mu_v^c-\lambda_v^c}$. This is an upper bound to $L_v^c$ since intermediate computations reduce the long term average number of packets in the system. The following gives a characterization of $L_v^c$.
\begin{prop}\label{FlowResult} 
The long-term average number of packets $L_v^c$ satisfies
\begin{align}
\label{FlowBounds}
\frac{\Hgf}{\mu_v^c-\Hgf}\leq L_v^c\leq M_v^c,\quad v\in V,\quad c\in C.
\end{align}
\end{prop}

\begin{proof}
In the case of no computation, $m_v^c=0$ and $n_v^c=L_v^c$, and the long-term average number of packets in $v$ satisfies $L_v^c=M_v^c$. This gives the upper bound.

When node $v$ computes it is true that $m_v^c>0$, and from (\ref{computationno}) $$L_v^c=\frac{\lambda_v^c}{\gamma_f(\lambda_v^c)}n_v^c\geq n_v^c=\frac{\gamma_f(\lambda_v^c)}{\mu_v^c-\gamma_f(\lambda_v^c)}.$$ We next provide an upper bound to the above equality where the bound can be approximated as function of $\Gamma_c$ as follows: 
\begin{align}
\label{Ubound}
L_v^c \leq \frac{\lambda_v^c}{\gamma_f(\lambda_v^c)} \frac{\gamma_f(\lambda_v^c)}{\lambda_v^c-\gamma_f(\lambda_v^c)}
\approx\frac{1}{1-\Gamma_c}=\frac{1}{1-\frac{\Hgf}{H(\X)}},
\end{align}
where $n_v^c\to 0$ as $\Gamma_c\to 0$, i.e., for deterministic functions where the gain of computing is the highest. On the other hand, as $\Gamma_c\to 1$, $n_v^c \to M_v^c$, i.e., the no computation limit.

The lower bound follows from Little's Law in (\ref{LittleComputation}):
\begin{align}
\label{Lbound}
L_v^c\overset{(a)}{\geq} \gamma_f(\lambda_v^c)\frac{2}{\mu_v^c-\gamma_f(\lambda_v^c)}\overset{(b)}{\geq}  \frac{\Hgf}{\mu_v^c-\Hgf},
\end{align}
where $(a)$ is because $\frac{1}{\lambda_v^c}\df$ is monotonically increasing with $m_v^c$ and not less than 
$\frac{1}{\mu_v^c-\gamma_f(\lambda_v^c)}$ (see Prop. \ref{stability}), and 
$(b)$ is due to the minimal rate required for recovering $f_c(\X)$. Manipulating the lower bound, we get 
$\gamma_f(\lambda_v^c)\geq 
\frac{\mu_v^c \Hgf}{2\mu_v^c-\Hgf}$. 
\end{proof}

Prop. \ref{FlowResult} yields a better IB than that of Slepian and Wolf \cite{SlepWolf1973} as the lower bound in (\ref{FlowBounds}) is no more than $\Hgf$ which we expect due to intermediate compression for computing.

\subsection{Optimizing Average Task Completion Time}
\label{mindelaynetwork}	

This section focuses on generalization to the case where we assume multiple classes of functions and allow conversion among classes, via exploiting the information-theoretic limits of function computation and flow conservation principles.

{\bf \em Processing factor vs arrival rate.}
The processing factor $\gamma_f(\lambda_v^c)$ satisfies $\forall c\in  C,\,\, v\in V\backslash\{s\}$ 
\begin{align}
\label{flow_entropy_relation_expanded}
\gamma_f(\lambda_v^c)\geq \lambda_v^c \Gamma_c,
\end{align} 
which follows from (\ref{flow_entropy_relation}), where $\lambda_v^c=\beta_v^c + \sum\nolimits_{w\in V'}\sum\nolimits_{c'\in C}\gamma_f(\lambda_{w}^{c'}) p_{w,v}^{rou}(c',c)$ using (\ref{throughputofclassc}), where $p_{v}^{rou}(c)=[p_{v,w}^{rou}(c,c')]_{w\in V',\,c'\in C}\in [0,1]^{(|V|-1)\times |C|}$ in (\ref{Proutingmulticlass}) and $\beta_v^c$ are known a priori. The relation in (\ref{flow_entropy_relation_expanded}) ensures that the flows routed from other nodes through the incoming edge of $v$ is at least $\lambda_v^c \Gamma_c$. Letting $\Prouc=[p_{v}^{rou}(c)]_{v\in V'}\in [0,1]^{(|V|-2)\times (|V|-1)\times |C|}$, ${\bm \beta}^c=[\beta_v^c]_{v\in V\backslash\{s\}}\in\mathbb{R}^{(|V|-1)\times 1}$, ${\bm \lambda}^c=[\lambda_v^c]_{v\in V\backslash\{s\}}\in\mathbb{R}^{(|V|-1)\times 1}$, ${\bm \mu}^c=[\mu_v^c]_{v\in V\backslash\{s\}}\in\mathbb{R}^{(|V|-1)\times 1}$, and ${\bm \gamma}_{f}({\bm \lambda}^c)=[\gamma_f(\lambda_v^c)]_{v\in V\backslash\{s\}}\in\mathbb{R}^{(|V-1|)\times 1}$, we can rewrite (\ref{flow_entropy_relation_expanded}) in vector form:
\begin{align}
\label{flow_entropy_relation_vector}
{\bm \lambda}^c\geq\gammaf &\geq {\bm \lambda}^c \Gamma_c=\left[{\bm \beta}^c+\Prouctilde \odot \gammaf\right]\Gamma_c\\
&=(I-\mbox{diag}(\Prouctilde)\Gamma_c)^{-1}{\bm \beta}^c \Gamma_c, \quad \forall c\in  C,\nonumber
\end{align}
where $\Prouctilde=\bm{1}_{|V|-2}^{\intercal}\Prouc \bm{1}_{|C|}\in\mathbb{R}^{(|V|-1)\times 1}$, $\bm{1}_{n}$ is a unit column vector of size $n$, $\odot$ is the 
Hadamard product, $I$ is an $(|V|-1)\times (|V|-1)$ identity matrix, and $\mbox{diag}(\Prouctilde)$ is an $(|V|-1)\times (|V|-1)$ matrix obtained by diagonalizing $\Prouctilde$. The range of ${\bm \lambda}^c$ can be computed as function of $\Gamma_c$ which is determined by the function's surjectivity and its computation cost. 
From (\ref{flow_entropy_relation_vector}) ${\bm \lambda}^c\geq {\bm \beta}^c+\mbox{diag}(\Prouctilde)  {\bm \lambda}^c\Gamma_c$. Furthermore, from (\ref{throughputofclassc}) ${\bm \lambda}^c\leq{\bm \beta}^c+\mbox{diag}(\Prouctilde)  {\bm \lambda}^c$. Combining these two yields the following lower and upper bounds, respectively: $$ (I-\mbox{diag}(\Prouctilde)\Gamma_c)^{-1} {\bm \beta}^c \leq {\bm \lambda}^c\leq(I-\mbox{diag}(\Prouctilde))^{-1}{\bm \beta}^c.$$ 

We next contrast 2 centralized formulations that optimize the {\em cost function} aggregated over $c\in C$, $v\in V$ in order to understand the savings of computation in a network setting.

\begin{figure}[t!]
\centering
\includegraphics[width=0.5\textwidth]{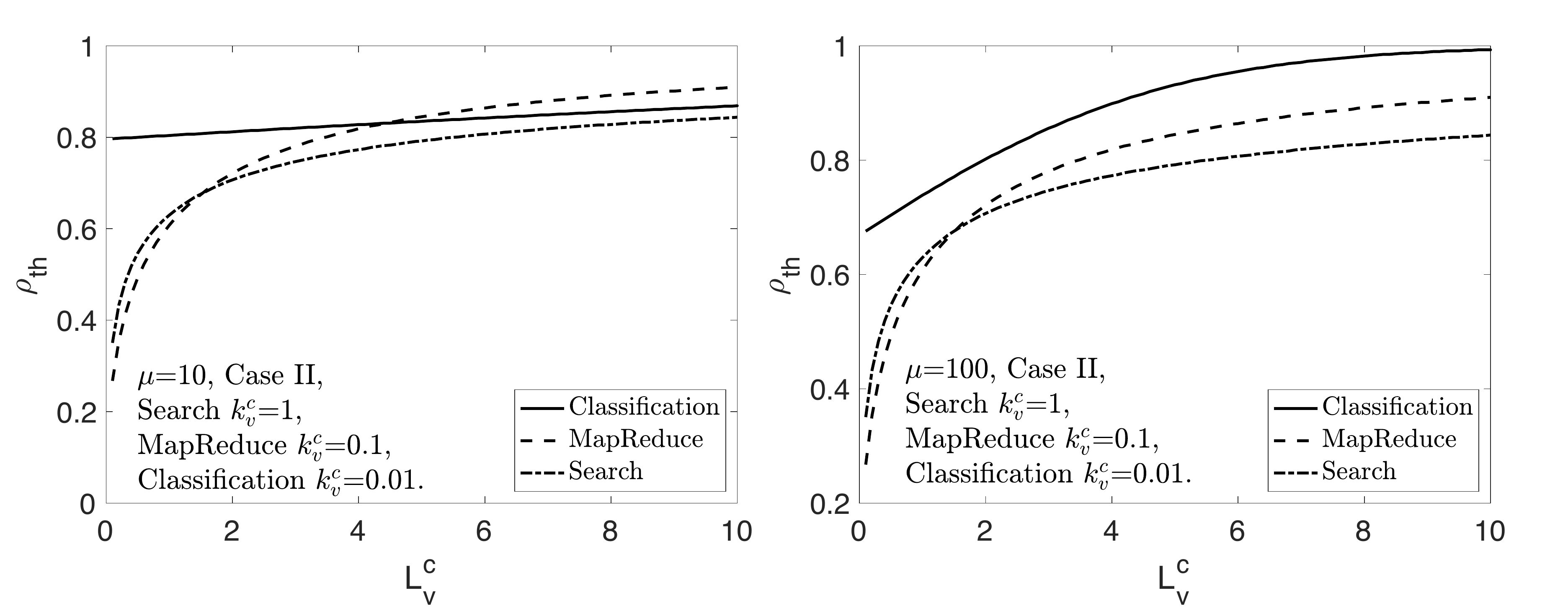}
\caption{The critical threshold $\rho_{th}$ for computation versus $L_v^c$ that satisfies Little's Law for computation $L_v^c=\gamma_f(\lambda_v^c) \cdot W_{v}^c$ in (\ref{LittleComputation}).} 
\label{thresholdisolated}
\end{figure}

\begin{figure*}[t!]
\centering
\includegraphics[width=\textwidth]{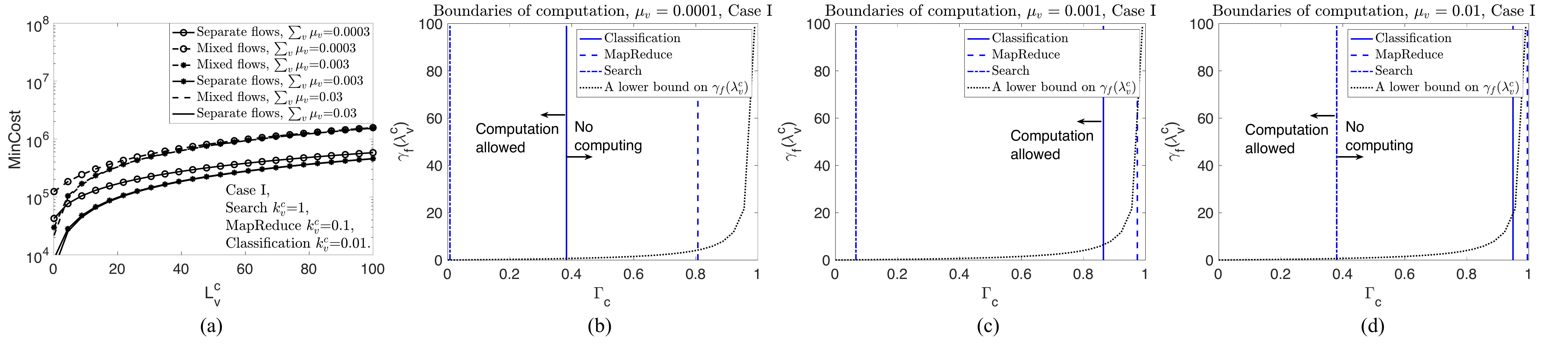}
\caption{Case I. Sublinear  surjection. 
(a) ${\rm MinCost}$ vs $L_v^c$. A lower bound on $\gamma_f(\lambda_v^c)$ for (b) $\mu_v^c=0.0001$, (c) $\mu_v^c=0.001$, (d) $\mu_v^c=0.01$. The boundaries of computation for each function class is indicated by vertical lines where computing 
is only allowed for $\Gamma_c$ below the 
line. } 
\label{CaseI_images}
\end{figure*}

{\bf \em Average cost of no computing.} 
We first formulate an optimization problem that relies on communicating the entire source data with no intermediate computing. This boils down to minimizing the communications cost 
and is formulated as 
\begin{equation}
\label{costoptimization_comms}
\begin{aligned}
{\rm CommsCost:} & \underset{\rho<1}{\min}
& & \sum\limits_{v\in V} \sum\limits_{c\in C} \ccomm(n_v^c),	
\end{aligned}
\end{equation}
where $\ccomm(n_s^c)$ and $\ccomm(n_0^c)$ are given in (\ref{cost_comms}), and $\ccomm(n_v^c)=\frac{1}{\mu_v^c-\lambda_v^c}$ for $v\in V'$. The solution to (\ref{costoptimization_comms}) is $${\rm CommsCost}=\sum\limits_{c\in C}\frac{1}{\bm{1}_{|V|-1}^{\intercal}({\bm \mu}^c-(I-\mbox{diag}(\Prouctilde))^{-1} {\bm \beta}^c)}.$$ 

{\bf \em Average cost of function computation.} 
We formulate a utility-based optimization problem by using the time complexities of computation tasks in (\ref{time_complexity_special_functions}) and decoupling the costs of communications and computation:
\begin{equation}
\label{costoptimization}
\begin{aligned}
{\rm MinCost:} & \underset{\rho,\, \sigma<1}{\min}
& & \sum\limits_{v\in V} \sum\limits_{c\in C} W_{v}^c\\	
& \hspace{0.3cm}\text{s.t.}
& & {\bm \mu}^c > {\bm \lambda}^c \geq \gammaf,\\ 
& & & \gammaf \geq {\bm \lambda}^c\Gamma_c,\quad \forall c\in  C,
\end{aligned}
\end{equation}		
where $W_v^c=\ccomp(m_v^c)+\ccomm(n_v^c)$ captures the total delay, and $\ccomp(m_v^c)$ and $\ccomm(n_v^c)$ were defined in (\ref{cost_compute}) and (\ref{cost_comms}). 
For this framework, in addition to the stability constraints given, we have to satisfy flow conservation, termination, generation, capacity and non-negativity constraints.
the relations between $\lambda_v^c$ and $\gamma_f(\lambda_v)$ are given in (\ref{throughputofclassc}).

{\bf \em Solution to ${\rm MinCost}$.} 
Given $\ccomm(n_v^c)$, $\ccomp(m_v^c)$, and $\df$ of different classes, we can solve for the optimal values of $\gamma_f(\lambda_v^c)$ and $\lambda_{v}^{c}$, $v\in V$, $c\in C$ that minimize ${\rm MinCost}$ in (\ref{costoptimization}). Then, using 
(\ref{flow_entropy_relation}), and by mapping the surjectivities $\Gamma_c$ to the class of functions, we can infer the type of flows (i.e., functions) that we can compute effectively. 
Even though ${\rm MinCost}$ may be non-convex and a global optimal solution may not exist or it might be NP-hard in some instances, to demonstrate the achievable gains in ${\rm MinCost}$ via computation over a communications only scheme we will run some experiments in Sect. \ref{performance} both by (i) separating and (ii) mixing different classes of flows (Search, MapReduce, and Classification) without precisely optimizing (\ref{costoptimization}).

{\bf \em Modeling classification via convex flow.}  	
Classification function $\ccomp(m_v^c)$ with a time complexity given in (\ref{time_complexity_special_functions}) is convex in $m_v^c$ which is linear and decreasing in $\gamma_f(\lambda_v^c)$, and $\ccomm(n_v^c)$ is convex in $\gamma_f(\lambda_v^c)$. Thus, ${\rm MinCost}$ is convex in $\gamma_f(\lambda_v^c)$. 
Hence, the optimal solution satisfies $\frac{\partial W}{\partial \gamma_f(\lambda_v^c)}=0$, $\forall v\in V$, $\forall c\in  C$. 
In the special case when $\ccomm$ dominates $\ccomp$, equality is required in the lower bound of in (\ref{flow_entropy_relation_vector}) such that $\gammaf = {\bm \lambda}^c\Gamma_c$ to minimize the convex ${\rm MinCost}$. Hence, $\lambda_{v}^{c}$'s can be extracted by solving the set of $|V|\times |C|$ equalities: $\gamma_f(\lambda_v^c)=\Big[\beta_v^c + \sum\nolimits_{w\in V'}\sum\nolimits_{c'\in C} \gamma_f(\lambda_w^{c'}) p_{w,v}^{rou}(c',c)\Big]\Gamma_c$, which yields $\gammaf= (I-\Prouctilde\Gamma_c)^{-1}{\bm \beta}^c \Gamma_c$, $c\in  C$.




{\bf \em Local solution of ${\rm MinCost}$.}
To find the local minima of ${\rm MinCost}$ for general computation cost functions, we can use the Karush-Kuhn-Tucker (KKT) approach in nonlinear programming \cite{Boyd2009}, provided that some regularity conditions are satisfied \cite{Bertsekas1999}. Allowing inequality constraints, KKT conditions determine the local optimal solution.
However, due to the lack of strong duality results for non-convex problems, ${\rm MinCost}$ can indeed be NP-hard in some instances \cite{manyem2010duality}.

Next in Sect. \ref{performance}, we will numerically evaluate the load threshold $\rho_{th}$ in (\ref{condition_computation}) for function classes with time complexities in (\ref{time_complexity_special_functions}), and the local solution of ${\rm MinCost}$ in (\ref{costoptimization}).

\section{Numerical Evaluation of Performance}
\label{performance}

Exploiting the computation network model, the cost model, and the flow analysis detailed in Sects. \ref{networkmodel}, \ref{costbreakdown}, \ref{flowanalysis}, our goal in this section is to explore how to distribute computation in the light of the fundamental limits provided in Sect. \ref{compresstocompute}.

{\bf \em Traffic intensity threshold $\rho_{th}$ for computing.} Leveraging Little's Law (\ref{LittleComputation}) for computation, we numerically evaluate the critical traffic intensity thresholds $\rho_{th}$ in (\ref{condition_computation}) of Prop. \ref{LoadThreshold} for the set of functions with time complexities $\df$ in (\ref{time_complexity_special_functions}). We contrast and illustrate their thresholds $\rho_{th}$ in Fig. \ref{thresholdisolated} for $\mu=10,\, 100$ for different cost models such that $\ccomp(m_v^c)=\frac{k_v^c}{\lambda_v^c} \df$ where $k_v^c>0$ is a constant proxy for modeling the cost as indicated in the figure. The threshold $\rho_{th}$ gives a range $\rho<\rho_{th}$ where computation is allowed and increases with $L_v^c$. 
Our experiment shows that $\rho_{th}$ is higher for Classification, implying that node can only compute if the flow is sufficient, and grows slower for low complexity Search and MapReduce functions, i.e., node computes even for small flow rates. This is because in the regime $\rho<\rho_{th}$, the growth of $\df$ is much faster than the communication cost of $\gamma_f(\lambda_v^c)$ for high complexity functions, rendering $L_v^c$ sufficiently small to ensure a valid threshold for computing. 
The benefit of computation is higher for MapReduce and even higher for Search. This implies that 
most of the resources are reserved for running simpler tasks to enable cost effective distributed computation.

{\bf \em Processing factor $\gamma_f(\lambda_v^c)$ vs entropic surjectivity $\Gamma_c$.} We next consider the network setting with mixing of flows. 
We numerically solve ${\rm MinCost}$ in (\ref{costoptimization}) for some special functions. We assume that $P^{rou}(c)$ is a valid stochastic matrix with entries chosen uniformly at random, and that ${\bm \beta}^c$ and ${\bm \mu}^c$ are known. Note that the values of $\Gamma_c$, $\lambda_v^c$, and hence $\gamma_f(\lambda_v^c)$ are coupled. It is intuitive that $\ccomp(m_v^c)$ increases with $\df$ for any function class, and the rate of increase is determined by $\lambda_v^c$. As $\Gamma_c$ increases, the nodes generate a higher $\gamma_f(\lambda_v^c)$ and $\ccomm(n_v^c)$ becomes higher. However, the exact behavior of ${\rm MinCost}$ is determined by the complex relationships between the flows. We investigate the behaviors of $\gamma_f(\lambda_v^c)$ and ${\rm MinCost}$ versus $\Gamma_c$ and how they change for different $\df$. 
Note that $\gamma_f(\lambda_v^c)$ increases in $\Gamma_c$ from (\ref{flow_entropy_relation}) because compression becomes harder as the entropy of the function $\Hgf$ grows. 
However, the connection between $\gamma_f(\lambda_v^c)$ and the time complexity $\df$ (and hence ${\rm MinCost}$) is not immediate. For example, while addition function has low complexity $\df$ and high entropy $\Hgf$, multiplication function has high complexity and low entropy.

\begin{figure*}[h!]
\centering
\includegraphics[width=\textwidth]{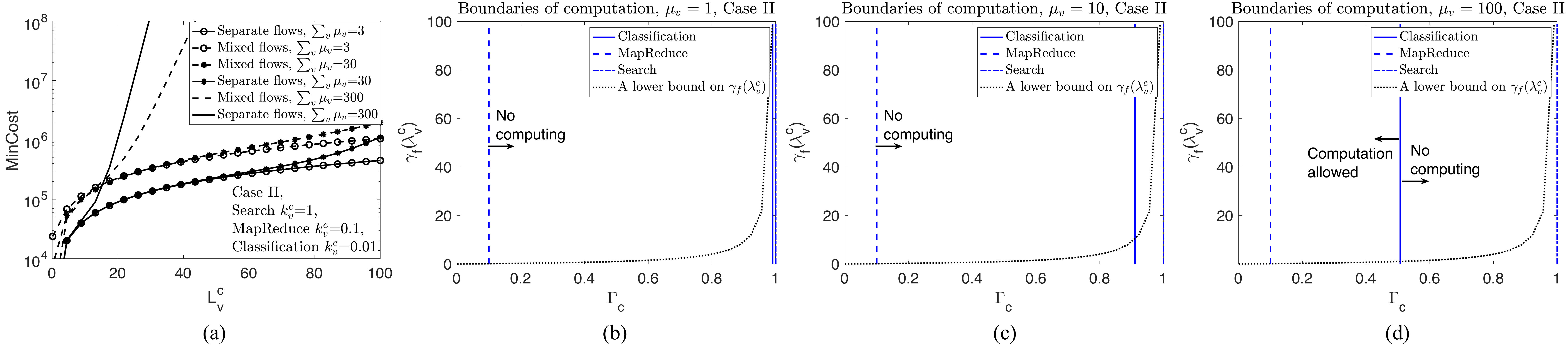}
\caption{Case II. Linear  surjection. 
(a) ${\rm MinCost}$ vs $L_v^c$. $\gamma_f(\lambda_v^c)$ for (b) $\mu_v^c=1$, (c) $\mu_v^c=10$, (d) $\mu_v^c=100$.} 
\label{CaseII_images}
\end{figure*}

We compare ${\rm MinCost}$ of (i) separating and (ii) mixing flows for Cases I and II, in Figs. \ref{CaseI_images} and \ref{CaseII_images}, respectively, with varying $\mu_v^c$. 
As shown in Fig. \ref{CaseI_images}-(a), as $\mu_v^c$ (processing power) increases ${\rm MinCost}$ decays in Case I. 
For Case I, ${\rm MinCost}$ decreases with $\mu_v^c$. 
The performance of (i)-(ii) similar when $\mu_v^c$ is very low and mixing flows performs worse. Because $\gamma_f(\lambda_v^c)$ sublinearly scales, $\ccomm(n_v^c)$ is decreases (convex), and $\ccomp(m_v^c)$ should decrease in $\mu_v^c$ because with increasing processing power $\lambda_v^c$ also increases and $\df$ decays because of Little's law when $v$ is considered in isolation $L_v^c$ decays when $W_v^c$ decays. 
From Fig. \ref{CaseII_images}-(a) for Case II ${\rm MinCost}$ increases in $\mu_v^c$. 
Because $\gamma_f(\lambda_v^c)$ linearly scales, $\ccomm$ is fixed, and $\ccomp$ increases in $\mu_v^c$ (due to $\ccomp(m_v^c)\geq \ccomm(n_v^c)$ as both $n_v^c$ and $m_v^c$ increase in $\mu_v^c$).  
%
{When $\mu_v^c$ is high $\ccomp(m_v^c)$ grows as $m_v^c$ increases. Therefore, if the flows are mixed among $V$ nodes, while it is true that the processed flow per node is higher versus when the flows are separated, the processing resources will be used more effectively instead of being utilized partially only for the assigned tasks. The sublinear surjection factor will help reduce $\ccomm(n_v^c)$. This will ensure at high $\mu_v^c$ that node $v$ can serve a higher $L_v^c$ without having to sacrifice ${\rm MinCost}$.} 
The savings of (ii) in ${\rm MinCost}$ are more apparent at high $L_v^c$ which maps to supporting a high intensity heterogeneous computing scenario via mixing rather than separate processing. 
Figs. \ref{CaseI_images} and \ref{CaseII_images}, (b)-(d) show the behavior of a lower bound on $\gamma_f(\lambda_v^c)$ in (\ref{flow_entropy_relation_vector}) as function of $\Gamma_c$, and indicate that the boundary is shifted towards right as $\mu_v^c$ increases, i.e., as the node's processing capability increases, a higher rate of surjectivity $\Gamma_c$ is allowed. As the communications delay becomes negligible, intermediate computations can improve ${\rm MinCost}$.
%

The behavior of ${\rm MinCost}$ is modified by the type and scaling of the computation cost as well as $\Gamma_c$. 
${\rm CommsCost}$ in (\ref{costoptimization_comms}) gives an upper bound to ${\rm MinCost}$. As the computation becomes more costly, the sources may only be partially compressed in order to optimize ${\rm MinCost}$. 
We also infer that there is no computing beyond some $\Gamma_c$ because allocating resources to computation no longer incurs less cost than ${\rm CommsCost}$. 

A node can perform computation and forward (route) the processed data if the range of $\Gamma_c$ that allows compression is flexible. This is possible when computation is cheap. However, if a node's compression range is small, i.e., when computation is very expensive, then the node only relays most of the time. While computing at the source and communicating the computation outcome might be feasible for some function classes, it might be very costly for some sets of functions due to the lack of cooperation among multiple sources. 
By making use of redundancy of data across geographically dispersed sources and the function to be computed, it is possible to decide how to distribute the computation in the network.

\section{Conclusions}
\label{conclusion}

We provided a novel perspective to function computation problem in networks. We introduced the novel notion of entropic surjectivity to assess the functional complexity. Extending Little's Law to computing, we derived regimes for which intermediate computations can provide significant savings. Our approach can be considered as an initial step towards understanding how to distribute computation and balance functional load in networks. 
Future directions include devising coding techniques for in network functional compression, by using compressed sensing and the compression theorem of Slepian and Wolf, employing the concepts of graph entropy, and exploiting function surjectivity. They also include more general network models beyond stationary and product-form.

\bibliographystyle{IEEEtran}
\bibliography{Derya}

\end{document}